\documentclass[lineno]{JFM-FLM_Au}

\usepackage{amsfonts,mathrsfs}
\newenvironment{proof}[1][Proof]{\par
  \noindent\textit{#1. }\ignorespaces
}{\hfill$\square$\par}

\usepackage{bm}
\usepackage{mathtools}
\usepackage{cleveref}
\newtheorem{theorem}{Theorem}
\newtheorem{proposition}{Proposition}
\newtheorem{lemma}{Lemma}

\usepackage{changes}
\usepackage{comment}

\usepackage{etoolbox}
\AtBeginEnvironment{equation}{\linenomath}
\AtBeginEnvironment{gather}{\linenomath}

\AtEndEnvironment{equation}{\endlinenomath}
\AtEndEnvironment{gather}{\endlinenomath}

\lefttitle{Zhang et al.}
\righttitle{Journal of Fluid Mechanics}

\title{Rapidly rotating internally heated convection: bounds on long-time averages}

\author{Yutong Zhang, Ali Arslan, Stefano Maffei \and Andrew Jackson}

\affiliation{Institute for Geophysics, ETH Zurich, Sonneggstrasse 5, Zurich 8092, Switzerland}

\corresau{Yutong Zhang, yutong.zhang@eaps.ethz.ch}

\begin{document}
\maketitle

\begin{abstract}
Convection on geophysical and astrophysical scales is subject to rapid rotation and strong heating from within the domain. In studying the long-time behaviour of the solutions for such a system, energy identities fail to capture the effects of rotation because the Coriolis force does no work, and rapid rotation can be prohibitive for direct numerical simulations. Instead, we derive an asymptotically reduced model for rapidly rotating convection driven by uniform internal heating between isothermal stress-free boundaries in a plane periodic layer. The main contribution is the proof of bounds on the mean temperature, and the mean vertical convective heat transport, in terms of the Rayleigh and Ekman numbers, in the limit of infinite Prandtl number. The first quantity measures thermal mixing, and the second the asymmetry in heat leaving the bottom and top boundaries due to convection, and unlike Rayleigh-B\'enard convection, the two are not a priori related. We employ alternative estimation techniques to those used in previous studies (Grooms \& Whitehead, 2014 \textit{Nonlinearity}, 28, 29) 
and identify two distinct scaling behaviours for both quantities. Finally, our bounds are optimised, within the methodology, and provide a rigorous constraint for future studies of rotation-dominated internally heated convection. 
\end{abstract}



\section{Introduction}

Convective flows in globally rotating reference frames are central to many problems in astrophysics and geophysics  \citep{roberts2007theory,miesch2005large}. 
In convective systems, buoyancy fluxes may originate at boundaries, as in canonical Rayleigh–Bénard convection (RBC), or be sustained by volumetric internal heating, as in internally heated convection (IHC). Although rotating RBC has been extensively characterised \citep{kunnen2021geostrophic,ecke2023turbulent}, rotating IHC has received limited investigation \citep{roberts1965thermal,jones2000onset,kaplan2017subcritical,lin2021large,hadjerci2024rapidly}, especially in the plane layer geometry \citep{arslan2024internally,ostilla2025rotationally}. In RBC it has been demonstrated that rapid rotation gives rise to the quasi-geostrophic (QG) flow regimes \citep{charney1948scale}, where the Coriolis force dominates buoyancy effects, balances the pressure gradient at leading order, and organises the flow into columnar structures aligned with the rotation axis. 
Above the rotation-suppressed onset of convection \citep{chandrasekhar2013hydrodynamic}, increasing buoyancy transitions the flow morphology from Taylor columns to plumes, and eventually to geostrophic turbulence before the buoyancy effects dominate \citep{sprague2006numerical,king2012heat,julien2012statistical,aguirre2022flow,song2024scaling}. 

Rapidly rotating IHC is relevant to the dynamics of liquid alloys in planetary cores driven thermally or compositionally by radiogenic heating, secular cooling \citep{nimmo2015energetics}, and the exsolution of light elements \citep{badro2016early,o2016powering}, as well as to subsurface oceans on icy moons that are subject to tidal heating \citep{nimmo2018icy}. In many of these contexts, direct numerical simulations and laboratory experiments cannot reach the appropriate parameter regimes \citep{cardin2015experiments,christensen2015numerical}. The principal difficulty lies in reaching the Ekman number ($E$) and Rayleigh number ($R$) required to reproduce geostrophic turbulence, where $E$ measures the ratio of the rotation period to the viscous diffusion timescale and $R$ the destabilising effect of heating relative to diffusion.
For instance, the Ekman number is approximately $10^{-15}$ in the Earth's liquid outer core \citep{de1998viscosity}, $10^{-13}$ in Mercury \citep{Seuren_2023} and possibly $10^{-15}$ in early Mars \citep{DIETRICH2013}. Given that convection requires a Rayleigh number at least as large as $E^{-4/3}$ \citep{chandrasekhar2013hydrodynamic,arslan2024internally}, the $R$ is at least on the order of $10^{17}$ and $10^{20}$ for these planetary bodies.
Due to the $E$-dependence of the maximum allowable time steps in computation \citep{julien2025rescaled}, and to mechanical strength limitations in laboratory experiments, studies of geostrophic turbulence are currently restricted to $E\gtrsim10^{-8}$ \citep{song2024scaling,cheng2020laboratory}. Moreover, the nonlinearity of the Navier–Stokes equations suggests that observations derived from a limited number of simulations may not accurately represent the systems' large-scale behaviour.

Rigorous analysis of the governing partial differential equations (PDEs) provides an alternative route to insight. An approach, first developed by Doering and Constantin, called the \textit{background field method} \citep{doering1992energy,doering1994variational,constantin1995variational,doering1996variational}, obtains bounds on long-time averages of diagnostic quantities in incompressible flows, determining how turbulent statistics depend on control parameters beyond the observable regime \citep{fantuzzi2022background}. The method relies on two key ideas. First, instead of optimising a mean quantity over exact PDE solutions, one considers a broader set of incompressible flows constrained only by integral balances such as energy and flux conservation. Second, the flow is decomposed into fluctuations and an arbitrary steady background field that satisfies the boundary conditions. By prescribing the geometry of this background field, the original variational problem over the full state space is reduced to an optimisation over the parameters of the background field. The resulting bounds depend on the chosen function and are valid, provided a quadratic form, termed the \textit{spectral constraint}, is non-negative. This method has been demonstrated to be a special case of the \textit{auxiliary functional method} \citep{chernyshenko2014polynomial,chernyshenko2022relationship} with quadratic functionals, in which bounded functionals encode the constraints from dynamics, and the background field serves as a perturbation of an energy norm.

In the present work, we employ {auxiliary functionals} to investigate a new model of rapidly rotating convection driven by internal heating. We consider a uniformly heated Boussinesq fluid confined between stress-free, isothermal horizontal parallel plates, with gravity directed vertically downward, parallel to the axis of rotation. 
The principal quantities characterising mean dissipation and energy transport at statistical equilibrium are the mean temperature 
and convective heat flux, each averaged over the volume and a sufficiently long time \citep{goluskin2016internally}.
These quantities are not \textit{a priori} related, in contrast to RBC, and can depend on $R$, $E$, as well as the Prandtl number $Pr$, defined as the ratio of viscous to thermal diffusivity, and the height to length ratio of the domain. Previous studies applying the auxiliary functional method have established rigorous bounds for non-rotating IHC \citep{lu2004bounds,whitehead2011internal,whitehead2012rigid,arslan2023rigorous,arslan2025new}. Bounds for rotating IHC have been derived under no-slip boundary conditions \citep{arslan2024internally}. However, if $E\lesssim R^{-2}$, the bounds on both mean temperature and mean convective heat flux are independent of $E$. As shown in table 1 of \citet{arslan2024internally}, bounds are established only for rotationally-affected IHC, which corresponds to transitional regimes between buoyancy and rotation dominated dynamics. Therefore, the $E$ and $R$ dependence in the limit of rapid rotation remains an open question. Moreover,
no investigation has addressed rotating IHC with stress-free boundaries. Here, we prove the first bounds on the mean temperature and convective heat flux with respect to $R$ and $E$ in the regime of rapid rotation at infinite $Pr$.

A standard energy identity of the Navier-Stokes equation fails to capture the effects of the Coriolis force and therefore such bounds contain no information on this crucial addition to the momentum balance compared with the non-rotating case. To progress, rather than working with the standard Navier-Stokes equations, we derive a model inspired by the \textit{non-hydrostatic quasi-geostrophic} (NH-QG) model \citep{julien1998new}, an asymptotically reduced PDE system valid for $E^{1/3}\ll 1$. The NH-QG equations were first developed by \citet{julien1998new} for rapidly rotating RBC between stress-free parallel plates, have been used to probe geostrophic turbulence \citep{sprague2006numerical,julien2012statistical,calkins2015asymptotic}, and show consistency with solutions of the full Navier–Stokes equations \citep{van2025bridging}. 
Rapid rotation suppresses vertical fluid motion and induces an anisotropy between the horizontal length scale ${L}$ and the vertical scale $H$, resulting in columnar structures and reducing the system’s effective dimensionality \citep{gallet2015exact}. Leveraging this rotational constraint, the NH-QG equations introduce an anisotropic non-dimensionalisation in which ${L}=E^{1/3} H$, a scaling motivated by linear stability theory \citep{chandrasekhar2013hydrodynamic,arslan2024internally}. The model describes dynamics on the small scale ${L}$ where the horizontal and vertical velocities are comparable, while assuming rotation to be sufficiently rapid so that the leading-order forces remain in geostrophic balance on this scale. The derivation of the NH-QG model proceeds by decomposing the flow fields into horizontally averaged components and fluctuations, expanding them in powers of $E^{1/3}$, and matching terms order by order. This procedure separates fast, small-scale fluctuations from the slow evolution of the mean fields, and results in a closed PDE system for the leading-order components of the vertical velocity $w$, the vertical vorticity $\zeta$, the horizontal stream function $\psi$, the horizontally averaged temperature $\Theta$, and the temperature fluctuation $E^{1/3}\theta$. 

The proof of bounds can be further aided by working in the infinite Prandtl number limit. Bounds in this limit are relevant for compositionally driven convection \citep{jones2015thermal}, and provide constraints that can be smoothly extended to finite $Pr$ cases \citep{wang2013bound,tilgner2022bounds}. With $Pr \to \infty$, the momentum equation reduces to a forced Stokes flow, which acts as a diagnostic constraint improving the estimates on the spectral constraint, the central step in deriving rigorous bounds within the auxiliary functional method. For the case of rapidly rotating RBC with stress-free boundaries, \citet{grooms2014bounds} were able to analyse the Green’s function of the Stokes flow to control the spectral constraint and proved that the Nusselt number $Nu$, which measures the ratio of the mean heat flux to its conductive counterpart, satisfies $Nu\lesssim R^3 E^{4}$ up to constants. Using a similar approach, \citet{pachev2020rigorous} demonstrated $Nu \lesssim R^2 E^2$ for no-slip boundary conditions, where the effects of Ekman pumping on the bulk flow are incorporated through modified boundary conditions \citep{julien2016nonlinear}. 
The present work provides the first application of NH-QG rescaling and bounding methods to rapidly rotating IHC systems.

Before presenting the main results, we introduce the coordinate system and averaging operators. An anisotropic Cartesian coordinate system $(x, y, Z)$ is adopted, where $Z\coloneqq E^{1/3}z$ is a rescaled vertical coordinate, with $z$ denoting the vertical coordinate measured in the same length units as $x$ and $y$.
The rescaled coordinate encodes the rotation-induced anisotropy between horizontal and vertical length scales \citep{julien1998new}, so that, as $\partial_Z =E^{-1/3}\partial_z$, the vertical and horizontal gradients of the flow fields are on the same order after rescaling. For any $f = f(x,y,Z,t)$ defined on the periodic domain $V=[0, L_x] \times [0, L_y] \times [0, 1]$, the following averaging operators are defined,
\begin{subequations}
\begin{gather}
    \overline{f}(Z,\,t)\coloneqq {\frac{1}{L_y L_x }\int_0^{L_y }\int_0^{L_x }f\,{\rm d}x\,{\rm d}y},\label{eq_ave_horizontal}\\
  \langle f \rangle_{V}(t)\coloneqq {\int_0^1  \overline{f}\,{\rm d}Z},\label{eq_ave_sp}\\
        \langle f\rangle_{V,\,t}\coloneqq\limsup_{\tau\to\infty}\frac{1}{\tau }{\int_0^\tau \langle f \rangle_{V} \,{\rm d}t}.\label{eq_ave_all}
\end{gather}
\end{subequations}
These represent the horizontal average, the volume average, and the long-time and volume average, respectively. The principal quantities to be bounded, namely the mean temperature and mean convective heat flux, are then denoted by $\langle \Theta \rangle_{V,\,t}$ and $\langle w \theta \rangle_{V,\,t}$. 

In this paper we prove the following two results. First, the lower bound for the mean temperature, 
\begin{equation}    \langle\Theta\rangle_{V,\,t}\gtrsim 0.66 \,E^{-1}R^{-3/4},
\end{equation}
up to second-order corrections, for $R\geq 19.2 E^{-4/3}$. We also show that the flow remains purely conductive, with $\langle\Theta\rangle_{V,\,t}=1/12$ for $R\leq 5.82 E^{-4/3}$, where the threshold provides a lower bound for the nonlinear critical Rayleigh number. For the intermediate range $5.82E^{-4/3}<R< 19.2 E^{-4/3}$, bounds are established as functions of $R$ and $E$ where terms with different scalings contribute at comparable orders.

The second result is the upper bound on the mean convective heat flux,
\begin{equation}
    \langle {w\theta}\rangle_{V,\,t}\lesssim 0.04E^{4/5}R^{3/5}-0.25,
\end{equation}
for $R\geq 30.2 E^{-4/3}$. Furthermore, we show that $\langle {w\theta}\rangle_{V,\,t}=0$ for $R\leq 3.50 E^{-4/3}$, and is bounded above by a multi-scaling function of $R$ and $E$ for $3.50 E^{-4/3}<R<30.2E^{-4/3}$.

 The remainder of the article is organised as follows. In \cref{preliminaries}, we introduce the reduced model for rapidly rotating internally heated convection, with the derivation of the governing equations detailed in appendix \ref{A_1_demonstration} within the NH-QG framework. After defining the system, we state rigorous bounds on $\langle\Theta\rangle_{V,\,t}$ and $\langle w\theta\rangle_{V,\,t}$ as theorems, and demonstrate the relevance of these quantities for characterising the dynamics. The lower bound for $\langle\Theta\rangle_{V,\,t}$ and the upper bound for $\langle w\theta\rangle_{V,\,t}$ are proved in sections \ref{Theta} and \ref{sec_wtheta}, respectively.
To complement the analysis, we present numerical optimisation results for both quantities in \cref{sec:numerical}.
Finally, \cref{Discussion} provides a discussion and summary of the present work. Appendix \ref{Appendix_V} establishes the boundedness of the auxiliary functional used in this study when evaluated along solutions of the governing equations, and appendix \ref{appendix_LST} presents a linear stability analysis of the conduction state. For notation, we use $\|\cdot\|_2$ and $\|\cdot\|_{\infty}$ for the $L^2$ and $L^\infty$ norms of a function over $Z\in[0,\,1]$, and $(\cdot)'$ denotes the differentiation with respect to $Z$.

\section{Preliminaries}
\label{preliminaries}
\subsection{Configuration and governing equations}
\label{sec_setup}

\begin{figure}[h!]
    \centering

    \includegraphics[width=0.7\linewidth]{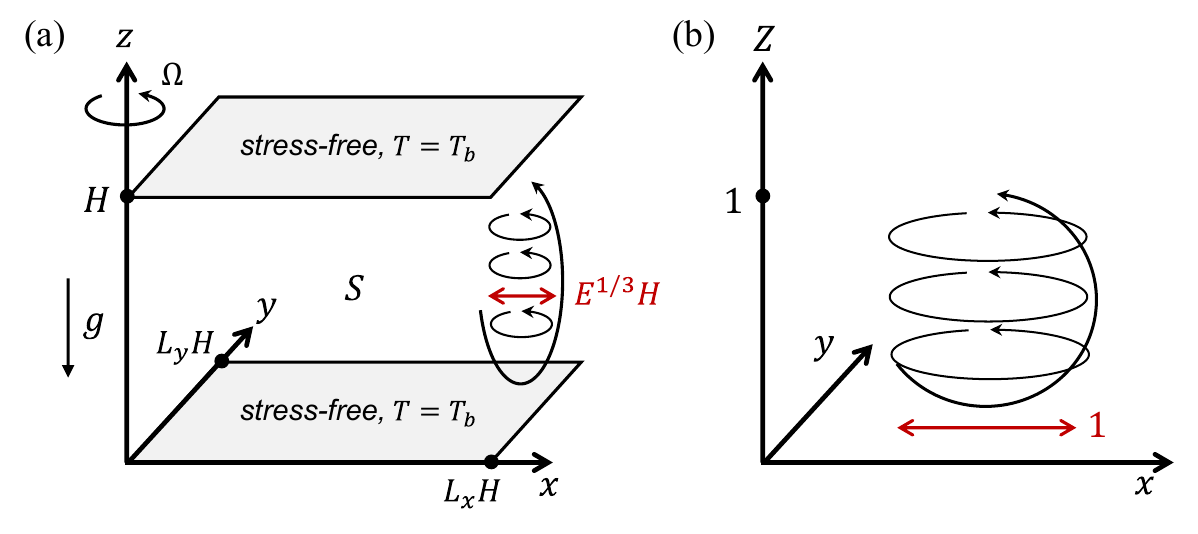}
    \caption{Left: a schematic diagram of rotating internally heated convection between two isothermal parallel plates with stress-free boundary conditions. The uniform heating source is denoted by $S$, and $H$ is the separation between the two plates in the direction of gravity, $g$. Horizontal periodicities are $L_x H$ and $L_y H$. The convection is subject to global rotation $\Omega$. Without loss of generality, the boundary temperature is set to $T_b = 0$. Circles indicate flow structures with a characteristic vertical scale $H$ and horizontal scale $E^{1/3}H$. Right: The same flow structures shown in dimensionless coordinates. The vertical and horizontal length units are $H$ and $E^{1/3}H$ respectively, with $E\ll 1$.}
    \label{fig_setup}
\end{figure}

As illustrated in figure \ref{fig_setup}(a), we consider an incompressible fluid between two parallel plates separated by a distance $H$ with horizontal periodicities of $L_x H$ and $L_y H$. The fluid is characterised by a reference density $\rho$, thermal diffusivity $\kappa$, kinematic viscosity $\nu$, specific heat $c_p$, and thermal expansivity $\alpha$. Convection is driven by a uniform volumetric heat source $S$ under a constant gravitational acceleration $-g \hat{\bm{z}}$. The system rapidly rotates about the vertical axis with constant angular velocity $\Omega$.

To describe rapid rotation, we use an asymptotically reduced set of equations to model the convection of a Boussinesq fluid in the limit of infinite Prandtl number, $Pr$. The Prandtl number is defined in \eqref{eq:nondim_parameters} and is the ratio of viscous to thermal diffusivity. The other control parameters are the Rayleigh $R$ and Ekman $E$ numbers, where $R$ quantifies the destabilising effect of heating to the stabilising effect of diffusion and $E$ the effect of inertia to rotation. These control parameters are defined as,  
\begin{equation}
    \label{eq:nondim_parameters}
    R = \frac{\alpha g S H^5}{\rho c_p \nu \kappa^2}, \quad E = \frac{\nu}{2\Omega H^2}, \quad Pr= \frac{\nu}{\kappa}.
\end{equation}
An additional key control parameter characterising the dynamics is the reduced Rayleigh number \citep{julien1998new} measuring the degree of supercriticality,
\begin{equation}\label{eq_gamma}
    \mathcal{R} = R E^{4/3}.
\end{equation}

To non-dimensionalise we choose the vertical length scale to be $H$, and the horizontal length scale to be $E^{1/3} H$. The effect of this anisotropic non-dimensionalisation on rotation-constrained columnar flow structures is illustrated in figure \ref{fig_setup}(b). The time and temperature scales are $E^{2/3} H^2 / \kappa$ and $S H^2/\rho c_p \kappa$ respectively. Then, in the asymptotic limit of rapid rotation the equations for convection at infinite $Pr$ (see appendix \ref{app:eqs} for a derivation) are
\begin{subequations}\label{eq_QG}
    \begin{gather}
        \partial_Z\psi ={\mathcal{R}}\theta+\nabla_h^2w,\label{eq_w}\\
        -\partial_Zw=\nabla_h^2\zeta,\label{eq_zeta}\\
        \zeta=\nabla_h^2\psi,\label{eq_psi}\\  \partial_t\theta +     J\left[\psi,\theta\right]+w\partial_Z \Theta=\nabla_h^2\theta,\label{eq_theta}\\ E^{-\frac{2}{3}}\partial_t \Theta+\partial_Z(\overline{w\theta})=\partial_Z^2 \Theta+1.\label{eq_Theta}
    \end{gather}
\end{subequations}
Here, $J[\psi,\,\cdot\,]\coloneqq \partial_x\psi\,\partial_y(\cdot) - \partial_y\psi\,\partial_x(\cdot)$, $\nabla_h(\cdot)\coloneqq \hat{\bm x}\partial_x(\cdot)+\hat{\bm y}\partial_y(\cdot)$, and $\nabla_h^2 \coloneqq \partial_x^2 + \partial_y^2$ denotes the horizontal Laplacian. Then, $\psi(x,y,Z,t)$ is the horizontal streamfunction, with the horizontal velocity being $-\nabla_h\times(\psi\hat{\bm{Z}})$; $w(x,y,Z,t)$ is the vertical velocity, $\zeta(x,y,Z,t)$ the vertical vorticity, $\Theta(Z,t)$ the horizontally averaged temperature field, and $E^{1/3}\theta(x,y,Z,t)$ the temperature anomaly relative to the horizontal average. It is noted that $w$, $\zeta$, $E^{1/3}\theta$, and $\Theta$ represent the leading-order non-zero components of the corresponding fields in an asymptotic expansion in powers of $E^{1/3}$, with the incompressibility condition satisfied by advection on higher-orders. If these variables are instead interpreted as the full fields, then \cref{eq_w,eq_zeta,eq_psi,eq_theta,eq_Theta} incur an error of $O(E^{1/3})$. The total temperature field  $T(x,y,Z,t)$ is approximated by
\begin{gather}\label{eq_totaltemperature}
     T=\Theta+E^{\frac{1}{3}}\theta.
\end{gather}
As demonstrated in appendix \ref{A_1_demonstration}, the horizontal means of $w$, $\zeta$ and $\psi$ vanish, while $\theta$ has zero horizontal mean by definition. Hence, $\overline{\theta} = \overline{w} = \overline{\zeta} = \overline{\psi} = 0$.

In the derivation of \eqref{eq_QG}, the geostrophic balance is assumed to constitute the leading-order force balance throughout the domain, thereby excluding any boundary-layer effects (see appendix \ref{A_1_demonstration}). For the physical consistency of the resulting PDEs, we impose stress-free mechanical boundary conditions. For the temperature, the isothermal boundary conditions $T=T_b$ are applied. Without loss of generality, $T_b$ is set to zero. The boundary conditions are thus given by
\begin{equation}\label{eq_BC}
        w=\partial_Z\zeta=\partial_Z\psi =\theta=\Theta=
        0,\quad\text{at }Z=0\text{ and }Z=1.
\end{equation}

\subsection{Main theorems}\label{sec_Theorems}

Having defined the governing equations and boundary conditions, here we state the two theorems proved in this work.

\begin{theorem}[Lower bound on mean temperature]\label{Theorem_Theta}
On the domain $V=[0,E^{-1/3}L_x]\times[0,E^{-1/3}L_y]\times[0,1]$, let $w$, $\psi$, $\theta:$ $V\times[0,\infty) \to \mathbb{R}$ and $\Theta:$ $[0,1]\times[0,\infty) \to \mathbb{R}$ satisfy the model defined in \eqref{eq_QG}, with periodic boundary conditions in $x$ and $y$,
and subject to the boundary conditions \eqref{eq_BC}. Let $\langle \Theta \rangle_{V,\,t}$ be the volume and long-time average of $\Theta$ as defined by \eqref{eq_ave_all}. Then,
    \begin{equation*}
     \langle\Theta\rangle_{V,\,t}\geq 
     \begin{cases}
     1.13E^{-8/7}R^{-6/7}-1.61 E^{-12/7}R^{-9/7},&5.82 E^{-4/3}\leq R<19.2E^{-4/3};\\
         0.66E^{-1}R^{-3/4}-
         1.51E^{-2}R^{-3/2},&  R\geq 19.2 E^{-4/3},    
     \end{cases}
\end{equation*}
where $5.82 E^{-4/3}$ corresponds to a lower threshold for nonlinear instability of the flow. For $R<5.82 E^{-4/3}$ the conductive state is globally asymptotically stable and $\langle\Theta\rangle_{V,\,t}=1/12$.
\end{theorem}
It is straightforward to verify that the lower bounds are positive for all $R$ over which the bounds hold.

\begin{theorem}[Upper bound on mean heat flux]\label{Theorem_wtheta}
On the domain $V=[0,E^{-1/3}L_x]\times[0,E^{-1/3}L_y]\times[0,1]$, let $w$, $\psi$, $\theta:$ $V\times[0,\infty) \to \mathbb{R}$ and $\Theta:$ $[0,1]\times[0,\infty) \to \mathbb{R}$ satisfy the model defined in \eqref{eq_QG}, with periodic boundary conditions in $x$ and $y$,
and subject to the boundary conditions \eqref{eq_BC}. Let $\langle w\theta \rangle_{V,\,t}$ be the volume and long-time average of $w\theta$ as defined by \eqref{eq_ave_all}. Then,
\begin{equation*}
     \langle {w\theta}\rangle_{V,\,t}\leq
     \begin{cases}
     \begin{aligned}
      0.01 E^{8/7}R^{6/7}-0.25
         +0.57 E^{-4/7}R^{-3/7}-0.37E^{-8/7}R^{-6/7}
         , &\\ 3.50E^{-4/3}\leq R<30.2E^{-4/3};&\\
    0.04E^{4/5}R^{3/5}-0.25- 0.15 E^{-4/5}R^{-3/5}+3.95E^{-8/5}R^{-6/5}, &\\   R\geq 30.2 E^{-4/3}.&
      \end{aligned}
     \end{cases}
\end{equation*}
For $R<3.50E^{-4/3}$, the conductive state is globally asymptotically stable and $\langle w\theta\rangle_{V,\,t}=0$.
\end{theorem}

\subsection{Quantities of interest}\label{sec_emergentquantities}

In this section, we demonstrate the relevance of the quantities for which we have stated bounds in \cref{Theorem_Theta} and \cref{Theorem_wtheta}. The mean temperature $\langle\Theta\rangle_{V,\,t}$ quantifies the thermal dissipation and effective mixing of the fluid, while the mean convective heat flux $\langle {w\theta}\rangle_{V,\,t}$ quantifies the asymmetry in heat flux out of the domain due to convection.
To see this, we multiply \cref{eq_theta,eq_Theta} by $\theta$ and $\Theta$ respectively, apply the spatial average \eqref{eq_ave_sp}, integrate by parts and sum to obtain
   \begin{equation}\label{eq_asV_energyevo}
    \frac{E^{-\frac{2}{3}}}{2} \frac{{\rm d}}{{\rm d} t}\langle \Theta^2+E^{\frac{2}{3}}\theta^2\rangle_{V}=-\langle \left|\nabla_h \theta\right|^2+(\partial_Z  \Theta)^2\rangle_{V}+\langle\Theta\rangle_{V},
\end{equation}
where the left-hand side represents the rate of change of the thermal energy density, while the first term on the right-hand side corresponds to the thermal dissipation rate. Further applying the long-time average yields the mean energy equilibrium
\begin{equation}
         \langle \,\left|\nabla_h \theta\right|^2+(\partial_Z \Theta)^2\rangle_{V,\,t}  = \langle\Theta\rangle_{V,\,t}.
        \label{eq_temperature_energy}
\end{equation}
On the other hand, multiplying \cref{eq_w,eq_zeta} by $w$ and $\psi$ respectively and proceeding similarly leads to the identity
\begin{equation}
    \langle\,\zeta^2+\left|\nabla_h w\right|^2\rangle_{V,\,t} = {\mathcal{R}}\langle w\theta\rangle_{V,\,t},\label{eq_kinetic_energy}
\end{equation}
with the left-hand side representing the dissipation arising from velocity gradients. As can be seen in \eqref{eq_temperature_energy} and \eqref{eq_kinetic_energy} the mean temperature and convective heat transport are non-negative, and this gives uniform lower bounds of $\langle  w \theta  \rangle_{V,\,t} \geq 0 $ and $ \langle \Theta \rangle_{V,\,t} \geq 0  $. For a uniform upper bound on  $\langle  w \theta  \rangle_{V,\,t}$, we multiply \eqref{eq_Theta} by $Z$ and take the average \eqref{eq_ave_all}. With $\mathcal{F}_T$ and $\mathcal{F}_B$ defined as the heat fluxes exiting through the top and bottom boundaries, this yields 
\begin{subequations}\label{eq_boundaryflux}
    \begin{gather}
        \mathcal{F}_T\coloneqq\left. -\limsup_{\tau \rightarrow \infty}\frac{1}{\tau}\int^{\tau}_0
        \partial_Z \Theta
        \right|_{Z=1} \mathrm{d}t =\frac{1}{2}+\langle w\theta\rangle_{V,\,t},\label{eq_FT}\\
        \mathcal{F}_B\coloneqq \left. \limsup_{\tau \rightarrow \infty}\frac{1}{\tau}\int^{\tau}_0  \partial_Z \Theta\right|_{Z=0} \mathrm{d}t =\frac{1}{2}-\langle w\theta\rangle_{V,\,t}.\label{eq_FB}
    \end{gather}
\end{subequations}
Given that the temperature within the domain is pointwise non-negative, $T \geq T_b=0$ with any deviation decaying exponentially in time \citep{arslan2021bounds}, the heat fluxes are also signed and thus we infer that $\langle w\theta \rangle_{V,\,t} \leq 1/2$. The uniform upper bound on the mean temperature is given by the conductive solution of the system \eqref{eq_QG} subject to \eqref{eq_BC}, where the conductive temperature profile is 
\begin{equation}\label{eq_conduction}
    \Theta_c=\frac{1}{2}Z(1-Z).
\end{equation}
It follows that $\langle \Theta \rangle_{V,\,t} \leq \langle \Theta_c \rangle_{V}=1/12$, with a proof analogous to that in \citet{goluskin2012convection}.


\subsection{Horizontal Fourier domain}

Before presenting the framework for deriving rigorous bounds on $\langle{\Theta}\rangle_{V,\,t}$ and $\langle {w\theta}\rangle_{V,\,t}$, we introduce a Fourier expansion in horizontal direction. Let $\mathcal{K}$ denote the set of all non-zero horizontal wavevectors $\bm{k}=k_x\hat{\bm{x}}+k_y\hat{\bm{y}}$ permitted by the dimensionless periods $E^{-1/3}L_x$ and $E^{-1/3}L_y$. For any scalar field denoted generically by $f$, we write
\begin{subequations}
\begin{gather}
    f=\overline{f}+\sum_{\bm{k}\in\mathcal{K}}{f}_{\bm{k}}(Z,\,t)\,{\rm e}^{{\rm i} \left(k_xx+k_yy\right)},\nonumber \label{eq_Fouriertransform}\\
    \mathrm{where} \quad \mathcal{K} = \left\{ (k_x,k_y) = E^{\frac13}\left(\frac{2\pi m}{L_x}, \frac{2\pi n}{L_y}\right)\Bigg| (m,n) \in \mathbb{Z}^2 \,\& \, (m,n)\neq(0,0) 
\right\}.
\end{gather}
\end{subequations}
Therein, ${f}_{\bm{k}}$ represent the Fourier coefficients. Since $f$ is real-valued, the coefficients satisfy ${f}_{\bm{k}}^{*}={f}_{-\bm{k}}$, with $(\cdot)^{*}$ denoting complex conjugation. Under the expansion \eqref{eq_Fouriertransform}, the combination of \cref{eq_psi,eq_w,eq_zeta} yields a diagnostic relation between $w_{\bm{k}}$ and $\theta_{\bm{k}}$,
\begin{equation}\label{eq_ODE}
    \theta_{\bm{k}}=\frac{1}{{\mathcal{R}}}\left(k^2 {w}_{\bm{k}}-\frac{1}{k^4}\partial_Z^2 {w}_{\bm{k}}\right),
\end{equation}
in which $k=\sqrt{k_x^2+k_y^2}$ is the magnitude of each wavevector.

\section{Bounds for mean temperature $\langle \Theta\rangle_{V,\,t}$}
\label{Theta}

\subsection{Bounding framework}\label{subsec_Theta_framework}

To bound $\langle\Theta\rangle_{V,\,t}$, we follow the {auxiliary functional method} \citep{fantuzzi2022background}, which is based on the observation that for any bounded functional $\mathcal{V}\{\Theta,\theta\}$ along the solutions of the governing equations \eqref{eq_QG}, the time derivative of  $\mathcal{V}\{\Theta,\theta\}$ has zero long-time average. Moreover, the time average of a quantity is bounded from below by its pointwise infimum, so that, where $\mathcal{V}$ is independent of spatial coordinates,
\begin{equation}\label{eq_infTheta}
\langle\Theta\rangle_{V,\,t}=\limsup_{\tau \rightarrow \infty}\frac{1}{\tau}\int^{\tau}_0 
\langle\Theta\rangle_{V}-\frac{{\rm d}{\mathcal{V}}}{{\rm d}t}\,\rm dt\geq \inf \left(\langle\Theta\rangle_{V}-\frac{{\rm d}{\mathcal{V}}}{{\rm d}t}\right)\geq B_{1}.
\end{equation}
Therefore, $B_{1}$, which may ultimately be shown to be a function of ${R}$ and $E$, serves as a lower bound for $\langle\Theta\rangle_{V,\,t}$, provided that there exists an auxiliary functional $\mathcal{V}\{\Theta,\theta\}$, such that  
\begin{equation}\label{eq_Q_Theta}
    Q\coloneqq \langle\Theta\rangle_{V}-\frac{{\rm d}\mathcal{V}}{{\rm d}t}-B_{1}\geq0.
\end{equation}
We consider the following quadratic auxiliary functional,
\begin{equation}\label{eq_auxiliaryfunction}
    \mathcal{V}\{\Theta,\theta\} = \frac{\beta}{2} \left\langle\theta^2+E^{-\frac{2}{3}}\Theta^2 \right\rangle_{V}-E^{-\frac{2}{3}}\left\langle\phi\,\Theta \right\rangle_{V},
\end{equation}
where $\beta>0$ is referred to as the balance parameter, and $\phi=\phi(Z)$ we call the background field, although in most prior work $\phi/\beta$ is typically referred to as the background field. The boundedness of $\mathcal{V}\{\Theta,\theta\}$ along solutions of \eqref{eq_QG} is demonstrated in appendix \ref{Appendix_V}. $\phi$ is piecewise differentiable with a square-integrable derivative, and satisfies the homogeneous boundary conditions
\begin{equation}\label{eq_phiBC}
    \phi(0)=0,\quad \phi(1)=0.
\end{equation}

Now, we show that condition \eqref{eq_Q_Theta} is enforced by a $B_{1}$ given by an appropriate choice of $\phi$, and that the resulting $B$ is a valid lower bound on $\langle\Theta\rangle_{V,\,t}$, provided that an associated quadratic form is non-negative.
\begin{proposition}[bound on $\langle\Theta\rangle_{V,\,t}$] \label{proposition_boundingTheta}

For a given background field $\phi:[0,1]\rightarrow \mathbb{R}$ satisfying the boundary conditions \eqref{eq_phiBC}, the scalar $B_{1}$ provides a lower bound on the mean temperature, where
\begin{equation}
    B_{1}\coloneqq\int_0^1 \phi-\frac{1}{4}\phi'^2\,{\rm d}Z,\label{eq_Program_Theta_OF}
\end{equation}
provided that the so-called \emph{spectral constraint} is also satisfied,
\begin{equation}\label{eq_Theta_SC}
    Q_{\bm{k}}\left\{w,\,\phi\right\}\geq0,\quad \forall k>0,\quad \forall \left.w \,\right| \,{w}(0)={w}(1)=0,
\end{equation}
in which 
\begin{equation}
    Q_{\bm{k}}\coloneqq k^6\|{w}_{\bm{k}}\|^2_2+2\|{w}'_{\bm{k}}\|^2_2+\frac{1}{k^6}\|{w}''_{\bm{k}}\|^2_2+{\mathcal{R}}\int_0^1\phi'\,\left(k^2\left|w_{\bm{k}}\right|^2-\frac{1}{k^4}\mathrm{Re}\left\{{w}''_{\bm{k}}{w}^*_{\bm{k}}\right\}\right) \, {\rm d}Z.\label{eq_Program_Theta_SC}
\end{equation}
\end{proposition}

\begin{proof}
From \cref{eq_theta,eq_Theta}, we start by calculating the Lie derivative of $\mathcal{V}$ in \eqref{eq_auxiliaryfunction}. By integration by parts, use of the boundary conditions \eqref{eq_BC} and fixing $\beta=1$ we get that,
\begin{equation}\label{eq_dVdt}
   \frac{{\rm d}\mathcal{V}}{{\rm d}t}  = -\left\langle |\nabla_h \theta|^2 + | \Theta'|^2 + \phi' w\theta + (Z- \phi') \Theta' + \phi \right\rangle_{V}.
\end{equation}
Proceed by substituting \cref{eq_dVdt} into the definition of $Q$ \eqref{eq_Q_Theta} and applying the horizontal Fourier expansion \eqref{eq_Fouriertransform} to $w$ and $\theta$. This yields the decomposition
\begin{equation}\label{eq_Q_decomposition}
    Q\{\Theta,w,\theta\}=Q_0\{\Theta\}+\sum_{\bm{k}\in\mathcal{K}}Q_{\bm{k}}\{w_{\bm{k}},\theta_{\bm{k}}\},
\end{equation}
where
\begin{subequations}
\begin{gather}
    Q_0=-B_{1}+\int_0^1\, |\Theta'|^2-\phi' \Theta' +\phi \,{\rm d}Z,\label{eq_Q1_0}\\
        Q_{\bm{k}}= k^2\| {\theta}_{\bm{k}}\|^2_2+\int_0^1  \phi' \,\mathrm{Re}\left\{ {w}_{\bm{k}} {\theta}_{\bm{k}}^*\right\}\, {\rm d}Z.\label{eq_Q1_k}
\end{gather}
\end{subequations}
Condition \eqref{eq_Q_Theta} is satisfied by requiring $Q_0 \geq 0$ and $Q_{\bm{k}} \geq 0 $ independently.
Noting that $Q_0$ is convex in $\Theta'$, we solve the associated Euler–Lagrange equation with respect to $\Theta$, subject to Dirichlet boundary conditions. The non-negativity condition on $Q_0$ is satisfied in the worst case by,
\begin{equation}\label{eq_Theta_infQ0}
    B_{1} = \inf_{  \substack{\Theta(Z) \\ 
    \Theta(0)= \Theta(1)=0 }  } \int_0^1\,  |\Theta' |^2-\phi'\Theta' +\phi\,\,{\rm d}Z,
\end{equation}
with the minimiser given by 
\begin{equation}\label{eq_Thetaopt_Theta}
     \Theta'_{opt}=\frac{\phi'}{2}.
\end{equation}
Then, substituting \eqref{eq_Thetaopt_Theta} into \eqref{eq_Theta_infQ0} gives the form \eqref{eq_Program_Theta_OF}.

Next, in \cref{eq_Q1_k} we eliminate the spectral coefficient $\theta_{\bm{k}}$ using the diagnostic relation \eqref{eq_ODE}, and rearrange through integration by parts to obtain
\begin{equation}
\label{eq:spectral_constraint_w}
\begin{aligned}
    Q_{\bm{k}}&=\frac{1}{{\mathcal{R}}^2}\left(k^6\|{w_{\bm{k}}}\|^2_2+2\|{w_{\bm{k}}'}\|^2_2+\frac{1}{k^6}\|{w_{\bm{k}}''}\|^2_2\right)\\
    &+\frac{1}{{\mathcal{R}}}\int_0^1\phi'\,\left(k^2\left|w_{\bm{k}}\right|^2-\frac{1}{k^4}\mathrm{Re}\left\{{w_{\bm{k}}''}{w_{\bm{k}}}^*\right\}\right) \, {\rm d}Z.
    \end{aligned}
\end{equation}
The condition $Q_{\bm{k}} \geq 0$ for all $\bm{k}$ is satisfied by the spectral constraint.
\end{proof}

Maximising \cref{eq_Program_Theta_OF} over all choices of $\phi$ subject to the spectral constraint yields the largest lower bound attainable within our framework. This optimisation problem is convex and, while analytically intractable, can be numerically solved. To obtain explicit bounds with rigorous parametrisation, here we focus on proving \cref{Theorem_Theta} by constructing a suboptimal ansatz for $\phi$, but present results of the numerical optimisation, where no specific form of $\phi$ is prescribed, in \cref{sec:numerical}. 

If $L_x$ and $L_y$ were prescribed, it would only be necessary for $Q_{\bm{k}}\geq0$ to hold for the corresponding subset of positive wavenumbers, in order for $B_{1}$ to be a valid bound. Here, we require $Q_{\bm{k}}\geq0$ for every positive $k$. The resulting bounds are thereby valid for arbitrary horizontal periods, and for flow structures with arbitrary horizontal length scales, provided that \eqref{eq_QG} remains representative for the dynamics.

\subsection{Ansatz}
We restrict $\phi(Z)$ to take the form 
\begin{equation}\label{eq_phi}
    \phi\coloneqq \begin{cases}
        \frac{A}{\delta^2}Z\left(2\delta-Z\right),  &0\leq Z\leq \delta,\\
        A,&\delta<Z<1-\delta,\\
       \frac{A}{\delta^2}\left(1-Z\right)\left(2\delta+Z-1\right), &1-\delta\leq Z\leq 1.
    \end{cases}
\end{equation}
Here $\phi(Z)$ is characterised by boundary layers of thickness $\delta\in\left(0,\,\frac{1}{2}\right]$ at the top and bottom of the layer, and by the amplitude parameter $A>0$ specifying the bulk value of $\phi$ (figure \ref{fig_quadraticBackground}). 

The choice of a background field with quadratic boundary layers and a constant interior is motivated by the following considerations. First, it renders the sign-indefinite integral in \cref{eq_Program_Theta_SC} identically zero in the bulk. Second, as will be demonstrated in \cref{sec_Theta_proof}, the maximal lower bound attainable with \cref{eq_phi} is exactly $1/12$ when ${\mathcal{R}}$ is below a certain threshold. This value is achieved with the choice $A=1/4$ and $\delta = 1/2$, for which the optimiser $\Theta_{opt}$ \eqref{eq_Thetaopt_Theta} is exactly equivalent to the conductive state \eqref{eq_conduction}, and the spectral constraint becomes equivalent to a sufficient condition for nonlinear stability based on the energy method. This is only possible by using the quadratic boundary layers in $Z$.

\begin{figure}
 \centering
    \includegraphics[width=0.45\linewidth]{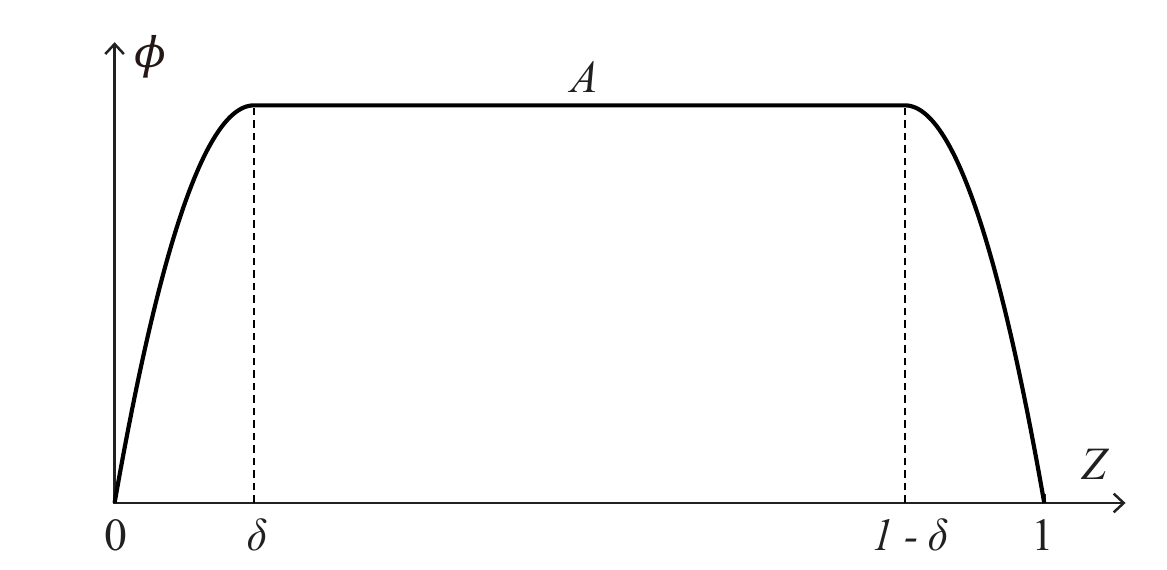}
    \caption{A sketch of the background profile \eqref{eq_phi} featuring quadratic boundary layers.}
    \label{fig_quadraticBackground}
\end{figure}

Substituting \cref{eq_phi} into \cref{eq_Program_Theta_OF,eq_Program_Theta_SC} leads to
\begin{equation}\label{eq_Theta_B}
    B_{1}=A\left(1-\frac{2\delta}{3}\right)-\frac{2A^2}{3\delta},
\end{equation}
along with
\begin{equation}\label{eq_Theta_Bk}
   Q_{\bm{k}}= \underbrace{ k^6\|{w}_{\bm{k}}\|_2^2+2\|{w}'_{\bm{k}}\|_2^2+\frac{1}{k^6}\|{w}''_{\bm{k}}\|_2^2 }_{P} +\frac{2{\mathcal{R}} A}{\delta^2}\left(I_B -I_T\right),  
\end{equation}
in which
\begin{subequations}\label{eq_Theta_Bk_decompose}
    \begin{gather}
        I_B= \int_0^\delta \left(\delta-Z \right)\left(
    k^2|{w}_{\bm{k}}|^2-\frac{1}{k^4}\mathrm{Re}\left\{{w}''_{\bm{k}}{w}^*_{\bm{k}} \right\}
    \right)\,{\rm d}Z,\label{eq_Ib}\\
    I_T=\int_{1-\delta}^{1} \left(Z-1+\delta\right)\left(
    k^2|{w}_{\bm{k}}|^2-\frac{1}{k^4}\mathrm{Re}\left\{{w}''_{\bm{k}}{w}^*_{\bm{k}} \right\}
    \right)\,{\rm d}Z.\label{eq_It}
    \end{gather}
\end{subequations}


\subsection{Sufficient conditions for the spectral constraint}\label{subsect_Theta_SCest}

As a central step in proving \cref{Theorem_Theta} using \cref{proposition_boundingTheta}, we establish in this section a closed-form sufficient condition on $\delta$ and $A$, from $\phi$ given by \cref{eq_phi}, that guarantees the spectral constraint  \eqref{eq_Theta_SC} is satisfied. This amounts to ensuring the non-negativity of \cref{eq_Theta_Bk} for all $k>0$. To this end, we first present two lemmas, each providing a sufficient condition for the non-negativity of \cref{eq_Theta_Bk}: one for the range $k\geq k_c$, and another for $0<k\leq k_c$, with $k_c>0$ being an arbitrary threshold wavenumber. Then, by reconciling the ranges of large and small wavenumbers we obtain a single condition to satisfy the spectral constraint. 


\subsubsection{For large wavenumbers $k\geq k_c$}

\begin{lemma}[Sufficient conditions for $Q_{\bm{k}}\geq 0$, large $k$] \label{lemma_Theta_largek}

Let $Q_{\bm{k}}$ be specified by \cref{eq_Theta_Bk,eq_Theta_Bk_decompose}, and $k_c>0$ be a lower threshold wavenumber. Suppose that $A>0$ satisfies
\begin{equation}\label{eq_SC_large}
    {\mathcal{R}} A\leq \sqrt{\frac{2}{5}}k_c.
\end{equation}
Then, $Q_{\bm{k}}\geq0$ for all $k\geq k_c$.
\end{lemma}

\begin{proof}
    We prove the lemma by constructing a lower bound on $Q_{\bm{k}}$ which is non-negative under the condition \eqref{eq_SC_large}. The contribution of $\phi$ from \eqref{eq_phi} to $I_B$ and $I_T$, as defined in \cref{eq_Ib,eq_It}, is symmetric. Hence, apart from deriving a lower bound on $P$, which has been defined in \eqref{eq_Theta_Bk}, we only prove upper and lower bounds on $I_B$, since the identical estimates hold for $I_T$. Using the boundary condition $w_{\bm{k}}(0)=0$ and integrating by parts, the $O(k^{-4})$ term in \cref{eq_Ib} can be written as
\begin{equation}
\begin{aligned}
    \int^{\delta}_0 (\delta - Z){\rm{Re}}\left\{ w''_{\bm{k}}w_{\bm{k}}^* \right\}  \mathrm{d}Z &= \int^{\delta}_0 \frac{1}{2} \left(\left|w_{\bm{k}}\right|^2 \right)' - (\delta -Z) |w'_{\bm{k}}|^2 \mathrm{d}Z \\
    &= \frac12 |w_{\bm{k}}(\delta)|^2 - \int^{\delta}_0 (\delta - Z )|w'_{\bm{k}}|^2 \mathrm{d}Z.
\end{aligned}    
\end{equation}
Then,
\begin{equation}\label{eq_Theta_IB_integral}
\begin{aligned}
    I_B &= \int^{\delta}_0 (\delta-Z) \left(k^2 |w_{\bm{k}}|^2 + \frac{1}{k^4} |w'_{\bm{k}}|^2 \right) \mathrm{d}Z -\frac{1}{2k^4}|w_{\bm{k}}(\delta)|^2 \\
    &\leq \int^{\delta}_0 (\delta-Z) \left(k^2 |w_{\bm{k}}|^2 + \frac{1}{k^4} |w'_{\bm{k}}|^2 \right) \mathrm{d}Z.
\end{aligned}    
\end{equation}
Next, using that $\left|w\right|\leq \|w\|_\infty$ and $\left|w'\right|\leq \|w'\|_\infty$, we get that
\begin{equation}\label{eq_est_It_large}
    I_B\leq \int_{0}^{\delta} \left(\delta-Z\right)\left(k^2\|w_{\bm{k}}\|_\infty^2+\frac{1}{k^4}\|w'_{\bm{k}}\|_\infty^2\right)\,{\rm d}Z=\frac{\delta^2}{2}\left(k^2\|w_{\bm{k}}\|_\infty^2+\frac{1}{k^4}\|w'_{\bm{k}}\|_\infty^2\right).
\end{equation}

On the other hand, the fundamental theorem of calculus implies
\begin{equation}\label{ineq_w_pointwise}
     \left|w_{\bm{k}}(Z)\right|=\left|\int_0^Z w'_{\bm{k}}(\eta)\,{\rm d}\eta \right|\leq\int_0^Z \|w'_{\bm{k}}\|_{\infty}\,{\rm d}\eta=Z\|w'_{\bm{k}}\|_{\infty},
\end{equation}
hence from the identity in \eqref{eq_Theta_IB_integral}, $I_B$ is lower bounded as
\begin{equation}\label{eq_est_Ib_large}
    I_B\geq -\frac{1}{2k^4}\left|w_{\bm{k}}(\delta)\right|^2\geq-\frac{\delta^2}{2k^4}\|w'_{\bm{k}}\|_{\infty}^2.
\end{equation}

Since $I_T$ is upper bounded by the right hand side of \eqref{eq_est_It_large}, inserting \eqref{eq_est_It_large} and \eqref{eq_est_Ib_large} into \cref{eq_Theta_Bk} yields a sufficient condition for $Q_{\bm{k}}\geq 0$,
\begin{equation}
    \label{eq:spec+const_mod}
    P - {\mathcal{R}} A \left(k^2\|w_{\bm{k}}\|_\infty^2+\frac{2}{k^4}\|w'_{\bm{k}}\|_\infty^2\right) \geq 0.
\end{equation}

To proceed we estimate $P$ from below. Denote $w_{\bm{k}}(Z)=w_{r}(Z)+\mathrm{i}w_{i}(Z)$, with $w_r$ and $w_i$ being real functions. As $w_{\bm{k}}(1)=w_{\bm{k}}(0)=0$, the mean value theorem guarantees the existence of points $Z_r$, $Z_i \in \left[0,\,1\right]$ such that ${w}'_{r}(Z_r) ={w}'_{i}(Z_i)=0$. Hence, for any $Z\in [0,\,1]$, the pointwise estimate holds that
\begin{equation}
\begin{aligned}
    \left|w'_{\bm{k}}(Z)\right|^2&=\left|w'_{r}(Z)\right|^2+\left|w'_{i}(Z)\right|^2\\
    &=2\int_{Z_r}^Z {w_{r}'}(\eta)w_{r}''(\eta)\,{\rm d}\eta+2\int_{Z_i}^Z {w_{i}'}(\eta)w_{i}''(\eta)\,{\rm d}\eta\\
    &\leq2\int_0^1 \left|w_r'(\eta)\right|\left|w_{r}''(\eta)\right|+\left|w_i'(\eta)\right|\left|w_{i}''(\eta)\right|  {\rm d}\eta.
    \end{aligned}
\end{equation}
Taking the supremum over $Z$ and applying the Cauchy–Schwarz inequality then gives
\begin{equation}\label{ineq_w'_inf}
    \begin{aligned}
    \|w'_{\bm{k}}\|_{\infty}^{2}&=\sup_{Z}\left|w_{\bm{k}}'(Z)\right|^2 \leq 2\left(\|w_r'\|_2\|w_r''\|_2+\|w_i'\|_2\|w_i''\|_2 \right)\\
    &\leq2\sqrt{\|w_r'\|_2^2+\|w_i'\|_2^2}\,\sqrt{\|w_r''\|_2^2+\|w_i''\|_2^2}\\
    &=2\|w_{\bm{k}}'\|_2 \|w_{\bm{k}}''\|_2. 
    \end{aligned}
\end{equation}
Similarly, one obtains
\begin{equation}\label{ineq_w_inf}
    \|w_{\bm{k}}\|_{\infty}^{2}\leq 2\|w_{\bm{k}}\|_2 \|w_{\bm{k}}'\|_2.
\end{equation}
Both of the above estimates are instances of the one-dimensional Gagliardo–Nirenberg interpolation inequality \citep{nirenberg1959elliptic}.

By Young’s inequality, together with inequalities \eqref{ineq_w'_inf} and \eqref{ineq_w_inf}, the quantity $P$ in \cref{eq_Theta_Bk} is bounded from below as
\begin{align}
    P&=k^6\|w_{\bm{k}}\|_2^2+\frac{2}{5}\|w'_{\bm{k}}\|_2^2+\frac{8}{5}\|w'_{\bm{k}}\|_2^2+\frac{1}{k^6}\|w''_{\bm{k}}\|_2^2 \\
    &\geq2k^3\sqrt{\frac{2}{5}}\|w_{\bm{k}}\|_2\|w'_{\bm{k}}\|_2+\frac{4}{k^3}\sqrt{\frac{2}{5}}\|w'_{\bm{k}}\|_2\|w''_{\bm{k}}\|_2 \nonumber \\
    &\geq \sqrt{\frac{2}{5}}\left(k^3\|w_{\bm{k}}\|_\infty^2+\frac{2}{k^3}\|w'_{\bm{k}}\|_\infty^2\right).\label{eq_est_P_large}
\end{align}

Next, substituting \eqref{eq_est_P_large} into \eqref{eq:spec+const_mod} gives the condition
\begin{equation}\label{eq_est_Bk_large}
    Q_{\bm{k}}\geq\left(\sqrt{\frac{2}{5}}k-{\mathcal{R}} A\right)\left(k^2\|w_{\bm{k}}\|_\infty^2+\frac{2}{k^4}\|w'_{\bm{k}}\|_\infty^2\right) \geq 0,
\end{equation}
which is satisfied for all $k\geq k_c$ if and only if \eqref{eq_SC_large} holds.
\end{proof}

\subsubsection{For small wavenumbers $0<k\leq k_c$}

\begin{lemma}[Sufficient conditions for $Q_{\bm{k}}\geq 0$, small $k$] \label{lemma_Theta_smallk}

Let $Q_{\bm{k}}$ be specified by \cref{eq_Theta_Bk,eq_Theta_Bk_decompose}, $k_c>0$ be an upper threshold wavenumber, and $a\in(0,\,1)$ be an arbitrary weight. Suppose that the positive pair $(A,\,\delta)$ satisfies
\begin{equation}\label{eq_SC_small}
    {\mathcal{R}} A\leq\min\left\{\frac{6a}{k_c^2\delta},\quad \sqrt{\frac{5}{k_c\delta}}\left(\frac{2-2a}{3}\right)^{\frac{1}{4}} \right\}.
\end{equation}
Then, $Q_{\bm{k}} \geq0$ for all $0<k\leq k_c$.
\end{lemma}

\begin{proof}
First, given that $w_{\bm{k}}(0)=0$, we can use the fundamental theorem of calculus and H\"older's inequality to write \eqref{ineq_w_pointwise} and also that
\begin{equation}
    \label{eq:simp_CS}
   \left| w_{\bm{k}}(Z)\right| =\left| \int^{Z}_0 w_{\bm{k}}'(\eta)\mathrm{d}\eta \right| \leq \sqrt{Z}\|w_{\bm{k}}' \|_2.
\end{equation}
Then, substituting \eqref{ineq_w_pointwise} and \eqref{eq:simp_CS} into \cref{eq_Ib}, integrating in $Z$ and using the Cauchy-Schwarz inequality gives,
\begin{align}
    \label{eq:IB_small}
    I_B &\leq k^2 \| w'_{\bm{k}}\|_2^2 \int^{\delta}_0 (\delta- Z)Z \mathrm{d}Z + \frac{1}{k^4} \| w'_{\bm{k}}\|_\infty\int^{\delta}_0 (\delta-Z)Z\, |w''_{\bm{k}}|\mathrm{d}Z  \nonumber \\
    &\leq \frac16 \delta^3k^2 \| w'_{\bm{k}}\|_2^2  + \frac{\delta^{5/2}}{\sqrt{30}\, k^4} \|w'_{\bm{k}} \|_\infty \|w''_{\bm{k}}\|_2.
\end{align}
By symmetry, $I_T$ is bounded from above by the right hand side of \eqref{eq:IB_small}. A similar procedure yields
\begin{equation}\label{eq_est_Ib_small}
    I_B\geq -\frac{1}{k^4} \| w'_{\bm{k}}\|_\infty \int^{\delta}_0 (\delta-Z)Z\, |w''_{\bm{k}}|\mathrm{d}Z \geq -\frac{\delta^{5/2}}{\sqrt{30}\, k^4} \|w'_{\bm{k}} \|_\infty \|w''_{\bm{k}}\|_2.
\end{equation}

Next, consider $P$. Let $a,\,b\in(0,\,1)$ be weighting parameters. By Young's inequality and \eqref{ineq_w'_inf}, the non-negative term in \cref{eq_Theta_Bk} can be written as
\begin{align}
\label{eq_est_P_small}
     P&=k^6\|w_{\bm{k}}\|_2^2+2a\|w'_{\bm{k}}\|_2^2+2(1-a)\|w'_{\bm{k}}\|_2^2+\frac{(1-b)}{k^6}\|w''_{\bm{k}}\|_2^2+\frac{b}{k^6}\|w''_{\bm{k}}\|_2^2 \nonumber\\
     &\geq 2a\|w'_{\bm{k}}\|_2^2+\frac{2}{k^3}\sqrt{2(1-a)(1-b)}\|w'_{\bm{k}}\|_2\|w''_{\bm{k}}\|_2+\frac{b}{k^6}\|w''_{\bm{k}}\|_2^2 \nonumber\\
     &\geq 2a\|w'_{\bm{k}}\|_2^2+\frac{1}{k^3}\sqrt{2(1-a)(1-b)}\|w'_{\bm{k}}\|_{\infty}^2+\frac{b}{k^6}\|w''_{\bm{k}}\|_2^2 \nonumber\\
     &\geq 2a\|w'_{\bm{k}}\|_2^2+\frac{2}{k^{9/2}}b^{\frac{1}{2}}(1-b)^{\frac{1}{4}}(2-2a)^{\frac{1}{4}}\|w'_{\bm{k}}\|_{\infty}\|w''_{\bm{k}}\|_2.
\end{align}
Here we fix $b=\frac{2}{3}$ to maximize the pre-factor $b^{\frac{1}{2}}(1-b)^{\frac{1}{4}}$. 
Substituting estimates \eqref{eq:IB_small} to \eqref{eq_est_P_small} into \cref{eq_Theta_Bk} then yields,
\begin{equation}\label{eq_est_Bk_small}
    Q_{\bm{k}}\geq \left(2a-\frac{1}{3} k^2\delta {\mathcal{R}} A\right)\|w'_{\bm{k}}\|_2^2+\frac{2}{k^{9/2}}\left[\left(\frac{2}{3}\right)^{\frac{3}{4}}(1-a)^{\frac{1}{4}}-2\sqrt{\frac{k\delta}{30}}{\mathcal{R}} A\right]\|w'_{\bm{k}}\|_{\infty}\|w''_{\bm{k}}\|_2.
\end{equation}
It suffices to establish $Q_{\bm{k}}\geq 0$ for all $k\in(0,\, k_c]$, provided both terms on the right-hand side of \eqref{eq_est_Bk_small} are non-negative on this interval. This requirement is equivalent to the condition \eqref{eq_SC_small}.
\end{proof}

\subsubsection{Reconciling the two ranges}

Aided by \cref{lemma_Theta_largek,lemma_Theta_smallk}, we are now in a position to state the condition for the spectral constraint to be satisfied at all wavenumbers.
\begin{lemma}[Sufficient conditions for spectral constraint]
\label{lemma_SC_Theta}
Let $\phi$ be defined by \cref{eq_phi}. Suppose that the positive pair $(A,\,\delta)$ in \cref{eq_phi} satisfies
\begin{gather}\label{eq_est_SC}
    {\mathcal{R}} A\leq c_{1}\,\delta^{-\frac{1}{3}}, \quad  \text{where}\quad c_{{1}}=\frac{\sqrt{5}}{3}\left(\frac{3}{10}\sqrt{557}-\frac{3}{2}\sqrt{5}\right)^{\frac{1}{3}}=1.156. 
\end{gather}
Then, $(A,\,\delta)$ satisfies the spectral constraint.
\end{lemma}

\begin{proof}
 If we equate the right hand sides of  \eqref{eq_SC_large} and \eqref{eq_SC_small}, then, solving the resulting two equations for $a$ and $k_c$ we get that,  
\begin{equation}\label{eq_kc}
    a=\frac{5\sqrt{2785}-125}{216},\quad  k_c=\frac{5}{3\sqrt{2}}\left(\frac{3}{10}\sqrt{557}-\frac{3}{2}\sqrt{5}\right)^{\frac{1}{3}}\delta^{-\frac{1}{3}}.
\end{equation}
Given \eqref{eq_kc}, both conditions \eqref{eq_SC_large} and \eqref{eq_SC_small}, for large and small $k$ respectively, become equivalent to \eqref{eq_est_SC}. Therefore, when \eqref{eq_est_SC} holds, then $Q_{\bm{k}}\geq 0$ for all $k>0$.  
\end{proof}

    In \citet{grooms2014bounds}, the authors prove the non-negativity of a spectral constraint identical to \eqref{eq_Program_Theta_SC}, with an alternative strategy to the one in \cref{subsect_Theta_SCest}. 
    In contrast to \citet{grooms2014bounds}, we do not examine the Green's function solution of the diagnostic equation \eqref{eq_ODE}. Instead, we have demonstrated that integral estimates directly on the quantities in $Q_{\bm{k}}$ suffice to prove non-negativity. 
    Even if our choice of $\phi(Z)$ was piecewise linear, like in \citet{grooms2014bounds}, our strategy would work the same.


\subsubsection{Optimality of the estimate}
    Before proving \cref{Theorem_Theta}, we demonstrate the optimality of the exponent $-1/3$ in \eqref{eq_est_SC} by establishing a necessary condition for the spectral constraint, using the single-mode velocity ansatz 
\begin{subequations}\label{eq_velocityansatz}
\begin{equation}
    {w}_{\bm{k}}^{ans}=\begin{cases}
        w_{k_a}^{ans}(Z), &\bm{k}=(\pm k_a,\,0),\\
        0,& \bm{k}\neq(\pm k_a,\,0).
    \end{cases}
\end{equation}
in which $k_a>0$ is an arbitrary wavenumber, and
\begin{equation}
    w_{k_a}^{ans}=\begin{cases}
        0,& 0\leq Z< 1-\delta,\\
        \frac{1}{\delta^5}\left(Z-1+\delta\right)^3\left(Z-1\right)\left(Z-1-\frac{\delta}{3}\right), & 1-\delta\leq Z\leq1.
    \end{cases}
\end{equation}
\end{subequations}

It is straightforward to check that $w_{k_a}^{\mathrm{ans}}$ satisfies the required boundary conditions and is continuously differentiable. Inserting this ansatz into \cref{eq_Theta_Bk} and rearranging gives
\begin{equation}\label{eq_Qk0ans}
    Q_{k_a}^{ans} =\frac{64}{105\delta}\left(\frac{1}{k_a^6\delta^2}+\frac{1}{12}+\frac{13}{4752}k_a^6\delta^2\right)-\frac{\mathcal{R}A}{27k_a\delta}\left(\frac{1}{k_a^3\delta}+\frac{19}{330}k_a^3\delta\right).
\end{equation}
The condition $Q_{k_a}^{ans}\geq0$ is necessary for the spectral constraint to be satisfied. Noting that
\begin{equation}\label{eq_Qk0ans_est}
    \frac{1}{k_a^6\delta^2}+\frac{1}{12}+\frac{13}{4752}k_a^6\delta^2<\frac{1}{k_a^6\delta^2}+\frac{7907}{75240}+\frac{13}{4752}k_a^6\delta^2=\left(\frac{1}{k_a^3\delta}+\frac{19}{330}k_a^3\delta\right)\left(\frac{1}{k_a^3\delta}+\frac{65}{1368}k_a^3\delta\right),
\end{equation}
substituting this estimate into \cref{eq_Qk0ans} and rearranging yields
\begin{equation}
    Q_{k_a}^{ans}
     < \frac{64}{105\delta}\left(\frac{1}{k_a^3\delta}+\frac{19}{330}k_a^3\delta\right)\left(\frac{1}{k_a^3\delta}+\frac{65}{1368}k_a^3\delta-\frac{35\mathcal{R}A}{576k_a}\right).
\end{equation}
A necessary condition for $Q_{k_a}^{ans}\geq0$ is the positivity of the sign-indefinite factor on the right hand side, namely 
\begin{equation}
    \mathcal{R}A<\frac{576}{35}\left(\frac{1}{k_a^2\delta}+\frac{65}{1368}k_a^4\delta\right).
\end{equation} 
We now specify $k_a=\left(684/65\right)^{1/6}\delta^{-1/3}$ to minimise the expression on the right hand side, which yields the following necessary condition for the spectral constraint,
\begin{equation}
         \mathcal{R}A<\frac{144}{7}\left(
     \frac{78}{475}\right)^{\frac{1}{3}}\delta^{-\frac{1}{3}}=11.27\,\delta^{-\frac{1}{3}}.
\end{equation}
Therefore, given the form of $\phi$ \eqref{eq_phi}, and assuming \textit{a priori} that $\delta$ decreases as ${\mathcal{R}}$ increases beyond a certain threshold, no sufficient condition for the spectral constraint can exist in the form ${\mathcal{R}} A\lesssim \delta^{-d}$ with the exponent $d>1/3$.


\subsection{Proof of theorem \ref{Theorem_Theta}}\label{sec_Theta_proof}

The final step toward \cref{Theorem_Theta} is to specify ${\mathcal{R}}$-dependent $A$ and $\delta$ that satisfy the condition of \cref{lemma_SC_Theta}. For algebraic simplicity, let
\begin{equation}\label{eq_R0}
    {\mathcal{R}}_{1}\coloneqq\frac{5\cdot2^{\frac{7}{3}}c_{1}}{3}=9.706
\end{equation}
be an \textit{a posteriori} threshold parameter, where $c_{1}$ is given in \cref{eq_est_SC}. We then choose
\begin{equation}\label{eq_choice}
    \delta=\frac{1}{2}\left(\frac{{\mathcal{R}}}{{\mathcal{R}}_{1}}\right)^{-\frac{3}{4}},\quad A=\frac{3}{20}\left(\frac{{\mathcal{R}}}{{\mathcal{R}}_{1}}\right)^{-\frac{3}{4}}.
\end{equation}
The requirement $\delta\leq 1/2$ imposes a lower threshold on ${\mathcal{R}}$ above which these choices are valid, namely ${\mathcal{R}}\geq {\mathcal{R}}_{1}$. Substituting \cref{eq_choice} into \cref{eq_Theta_B} yields the explicit expression
    \begin{equation}\label{eq_bound_Theta}
     B_{1}= \frac{3}{25}\left(\frac{{\mathcal{R}}}{{\mathcal{R}}_{1}}\right)^{-\frac{3}{4}}-\frac{1}{20}\left(\frac{{\mathcal{R}}}{{\mathcal{R}}_{1}}\right)^{-\frac{3}{2}},\quad \forall{\mathcal{R}}\geq {\mathcal{R}}_{1}.
     \end{equation}
Since the choice of $\delta$ and $A$ satisfy the condition of \cref{lemma_SC_Theta}, from \cref{proposition_boundingTheta} $B_1$ provides a lower bound on $\langle\Theta\rangle_{V,\,t}$. For small $\mathcal{R}$, the unconstrained global maximum of \cref{eq_Theta_B}$, 1/12$, is attained at $\delta=1/2$ and $A=1/4$. This combination is feasible for ${\mathcal{R}}\leq 3{\mathcal{R}}_{1}/5$, as determined by \eqref{eq_est_SC}. Since $1/12$ is also the uniform upper bound, it follows that $\langle \Theta\rangle_{V,\,t}=1/12$ in this regime, and the feasibility condition provides a sufficient condition for nonlinear stability of the conductive state. For ${\mathcal{R}}>3{\mathcal{R}}_{1}/5$, the spectral constraint \eqref{eq_est_SC} becomes active, and we choose
\begin{equation}\label{eq_choice_2}
    \delta=\sqrt{A}=\frac{1}{2}\left(\frac{5{\mathcal{R}}}{3 {\mathcal{R}}_{1}}\right)^{-\frac{3}{7}},
\end{equation}
which leads to the bound
\begin{equation}\label{eq_bound_Theta_2}
    B_{1}= \frac{1}{4}\left(\frac{5{\mathcal{R}}}{3{\mathcal{R}}_{1}}\right)^{-\frac{6}{7}}-\frac{1}{6}\left(\frac{5{\mathcal{R}}}{3{\mathcal{R}}_{1}}\right)^{-\frac{9}{7}},\quad \forall \mathcal{R}\geq \frac{3}{5}\mathcal{R}_{1}.
\end{equation}
It remains effective up to ${\mathcal{R}}=1.976\,{\mathcal{R}}_{1}$, beyond which \cref{eq_bound_Theta} provides a sharper analytical bound. Theorem \ref{Theorem_Theta} is therefore proven by substituting the value of $\mathcal{R}_1$. \hfill $\square$

As demonstrated in appendix \ref{appendix_LST}, the conductive state \eqref{eq_conduction} is linearly stable for $\mathcal{R}\leq \mathcal{R}_c=71.75$, hence the range of validity for \cref{eq_bound_Theta} covers values of ${\mathcal{R}}$ above the linear onset of convection. However, although a lower bound is provided by $\frac{3}{5}\mathcal{R}_1$, the nonlinear critical $\mathcal{R}$ remains unknown, and may be significantly lower than the linear threshold. For instance, in nonrotating internally heated convection with stress-free and isothermal boundary conditions, the nonlinear critical Rayleigh number can be 40\% smaller than the linear critical value \citep{goluskin2016internally}. Motivated by this, we have also proved bounds valid for ${\mathcal{R}} < \mathcal{R}_c$. Lower bounds for $\langle\Theta\rangle_{V,\,t}$ across different ranges of ${\mathcal{R}}$ are presented in figure \ref{fig_Theta_bounds}. Ideally, one seeks the maximal lower bound, yet the maximisation of \cref{eq_Theta_B} under \eqref{eq_est_SC} can only be solved numerically. By using \cref{eq_bound_Theta,eq_bound_Theta_2}, we obtain analytical bounds that are at most 2.4\% below the numerical maximum, denoted by $B^{num}_{1}$.

For $\mathcal{R}\geq\mathcal{R}_{1}$, $A$ and $\delta$ are chosen as \eqref{eq_choice} so that the bound asymptotically approaches the maximal lower bound $B^{num}_{1}$, although the QG regimes extend only to finite $\mathcal{R}$. This is shown as follows. By construction, for $\mathcal{R}\geq\mathcal{R}_{1}$ we have
\begin{equation}
    B^{num}_{1}=\max_{\eqref{eq_est_SC},\,\delta\leq1/2} A\left(1-\frac{2\delta}{3}\right)-\frac{2A^2}{3\delta} <\max_{\eqref{eq_est_SC},\,\delta\leq1/2} A-\frac{2A^2}{3\delta}=\frac{3}{25}\left(\frac{\mathcal{R}}{\mathcal{R}_{1}}\right)^{-\frac{3}{4}},
\end{equation}
with the combination \eqref{eq_choice} being the maximiser. Hence, compared with \cref{eq_bound_Theta},
\begin{equation}
    B_{1}<B^{num}_{1}<\left[1+O(\mathcal{R}^{-3/4})\right]B_{1},
\end{equation}
which implies that \cref{eq_bound_Theta} is equivalent to the maximal lower bound on leading order. For $\mathcal{R}\leq\mathcal{R}_{1}$, $A$ and $\delta$ are chosen as \eqref{eq_choice_2} so that the bound is tangent to $B^{num}_1$ at $3\mathcal{R}_{1}/5$.

In \cref{eq_bound_Theta,eq_bound_Theta_2}, the prefactors of terms are determined by $\phi$, while $\mathcal{R}_{1}$ arises from the particular estimates used to bound $Q_{\bm{k}}$ from below. As such, ${\mathcal{R}}_1$ is a product of the methodology and can be improved if sharper estimates on $Q_{\bm{k}}$ are possible.

\begin{figure}[h]
    \centering
    \includegraphics[width=0.8\linewidth]{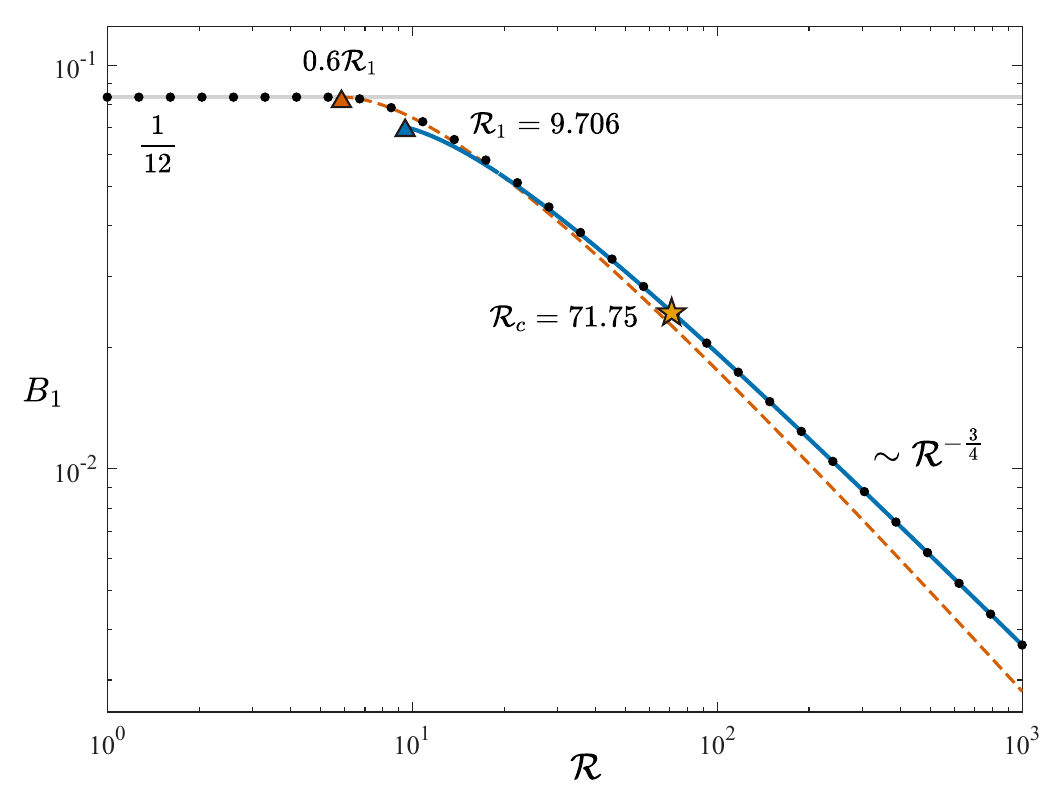}
    \caption{Lower bounds on $\langle \Theta \rangle_{V,\,t}$ from \eqref{eq_bound_Theta} (blue solid line) and \eqref{eq_bound_Theta_2} (red dashed line), compared with the uniform upper bound of $1/12$ (grey solid line) and the numerical maximum (black dots) of \eqref{eq_Theta_B} subject to \eqref{eq_est_SC} and the condition $\delta \leq 1/2$. Triangle markers denote the ${\mathcal{R}}$ above which \eqref{eq_bound_Theta} and \eqref{eq_bound_Theta_2} are valid, with the red triangle corresponding to a lower threshold of nonlinear convective instability, and the star indicating the linear onset of convection.}
    \label{fig_Theta_bounds}
\end{figure}

\section{Bounds for mean convective heat flux $\langle {w\theta}\rangle_{V,\,t}$}
\label{sec_wtheta}

\subsection{Bounding framework}

The same method as that outlined in \cref{subsec_Theta_framework} is used to bound the mean convective heat flux $\langle {w\theta}\rangle_{V,\,t}$ from above. Since the time average of a quantity does not exceed the pointwise supremum, we have that
\begin{equation}\label{eq_supwtheta}
\langle{w\theta}\rangle_{V,\,t}=\limsup_{\tau \rightarrow \infty}\frac{1}{\tau}\int^{\tau}_0 \langle{w\theta}\rangle_V+\frac{{\rm d}\mathcal{V}}{{\rm d}t}\,{\rm d}t\leq \sup \left(
\langle {w\theta}\rangle_V+\frac{{\rm d}\mathcal{V}}{{\rm d}t}
\right)\leq B_{2},
\end{equation}
where $\mathcal{V}$ is defined by \eqref{eq_auxiliaryfunction}. It thereby suffices for $B_{2}$ to be an upper bound on $\langle{w\theta}\rangle_{V,\,t}$ provided that
\begin{equation}\label{eq_Q_wtheta}
     Q\coloneqq B_{2}-
         \langle{w\theta}\rangle_{V}-\frac{{\rm d}\mathcal{V}}{{\rm d}t}\geq0.
\end{equation}
On this basis, we state the problem.

\begin{proposition}[bound on ${\langle {w\theta}\rangle_{V,\,t}}$] \label{proposition_boundingwtheta}
Let $\phi:[0,1] \rightarrow \mathbb{R}$ satisfy the boundary conditions \eqref{eq_phiBC}, and let $\beta>0$, then the upper bound $B_{2}$ on $\langle w \theta \rangle_{V,\,t}$ is given by
\begin{equation}
    B_{2}\coloneqq\frac{\beta}{48}+\int_0^1 \frac{1}{4\beta}{\phi'}^2-\frac{1}{2}\phi  
    \,{\rm d}Z,\label{eq_Program_wtheta_OF}
\end{equation}
provided that $\phi$ and $\beta$ further satisfy the \emph{spectral constraint}:
\begin{equation}\label{eq_wtheta_SC}
   Q_{\bm{k}}\left\{w,\,\phi,\,\beta\right\}\geq0,\quad \forall k>0,\quad \forall \left.w\,\right| \,{w}(0)={w}(1)=0,
\end{equation}
in which 
\begin{equation}
    Q_{\bm{k}}\coloneqq k^6\|{w}_{\bm{k}}\|^2_2+2\|{w}'_{\bm{k}}\|^2_2+\frac{1}{k^6}\|{w}''_{\bm{k}}\|^2_2+\frac{{\mathcal{R}}}{\beta }\int_0^1\left(\phi'-1\right)\left(k^2\left|w_{\bm{k}}\right|^2-\frac{1}{k^4}\mathrm{Re}\left\{{w}''_{\bm{k}}{w}^*_{\bm{k}}\right\}\right)\, {\rm d}Z.\label{eq_Program_wtheta_SC}
\end{equation}
\end{proposition}

\begin{proof}
    Similar to the demonstration of \cref{proposition_boundingTheta}, inserting \cref{eq_dVdt} into \cref{eq_Q_wtheta}, and expanding $w$ and $\theta$ in horizontal Fourier modes leads to the decomposition of $Q$ in \cref{eq_Q_decomposition}, and the condition $Q\geq 0$ becomes two independent conditions,
    \begin{subequations}
\begin{gather}
        Q_0=B_{2}+\int_0^1
    \beta\,|\Theta' |^2-\left(\phi'-\beta Z \right)\Theta'+\phi
    \,\,{\rm d}Z \geq 0,\label{eq_Q2_0}\\
    Q_{\bm{k}}=\beta k^2\|{\theta}_{\bm{k}}\|^2_2+\int_0^1\left(\phi'-1\right)\mathrm{Re}\left\{{w}_{\bm{k}}{\theta}_{\bm{k}}^*\right\}\, {\rm d}Z \geq 0.\label{eq_Q2_k}
\end{gather}
\end{subequations}
Since $\beta$ is positive, the integrand in $Q_0$ is convex in $\Theta'$. Then the worst case value of $B_{2}$ is given by
\begin{equation}\label{eq_wtheta_infQ0}
   B_{2} =  \sup_{  \substack{\Theta(Z)\\\Theta(0)= \Theta(1)=0 }  }{\int_0^1
    -\beta\,|\Theta'|^2+\left(\phi'-\beta Z \right)\Theta'-\phi
    \,\,{\rm d}Z},
\end{equation}
the maximiser of which is solved as
\begin{equation}\label{eq_Thetaopt_wtheta}
   \Theta'_{opt}=\frac{\phi'}{2\beta} - \frac{Z}{2} + \frac14.
\end{equation}
Substituting $\Theta'_{opt}$ back into $B_2$ yields the desired result.

Using the diagnostic relation \eqref{eq_ODE} and integrating by parts, $Q_{\bm{k}}$ can be equivalently written as
\begin{equation}
\begin{aligned}
    Q_{\bm{k}}&=\frac{\beta}{{\mathcal{R}}^2}\left(k^6\|{w}_{\bm{k}}\|^2_2+2\|{w_{\bm{k}}'}\|^2_2  
    +\frac{1}{k^6}\|{w_{\bm{k}}''}\|^2_2\right)\\
    &+\frac{1}{{\mathcal{R}} }\int_0^1\left(\phi'-1\right)\left(k^2\left|w_{\bm{k}}\right|^2-\frac{1}{k^4}\mathrm{Re}\left\{{w_{\bm{k}}''}{w_{\bm{k}}}^*\right\}\right)\, {\rm d}Z,
\end{aligned}    
\end{equation}
the non-negativity of which for all admissible $\bm{k}$ and $w_{\bm{k}}$ is satisfied by the spectral constraint. Hence, the condition \eqref{eq_Q_wtheta} is satisfied by $B_{2}$ as given in \cref{eq_Program_wtheta_OF}, together with the spectral constraint.
\end{proof}


\subsection{Ansatz}

For the proof of \cref{Theorem_wtheta}, a different family of background fields, $\phi$, from the previous section \eqref{subsec_Theta_framework} is used.
The background field is characterised by a quadratic top boundary layer, and is linear in the bulk with  $\phi'=1$,
\begin{equation}\label{eq_phi_wtheta}
    \phi= \begin{cases}
        Z,&0\leq Z<1-\delta,\\
       \frac{1}{\delta^2}\left(1-Z\right)\left(2\delta-\delta^2+Z-1\right), &1-\delta\leq Z\leq 1.
    \end{cases}
\end{equation}
As illustrated in figure \ref{fig_quadraticBackground2}, $\phi$ is fully specified by the boundary layer thickness $\delta\in\left(0,\,1\right]$. 

\begin{figure}
\centering
    \includegraphics[width=0.45\linewidth]{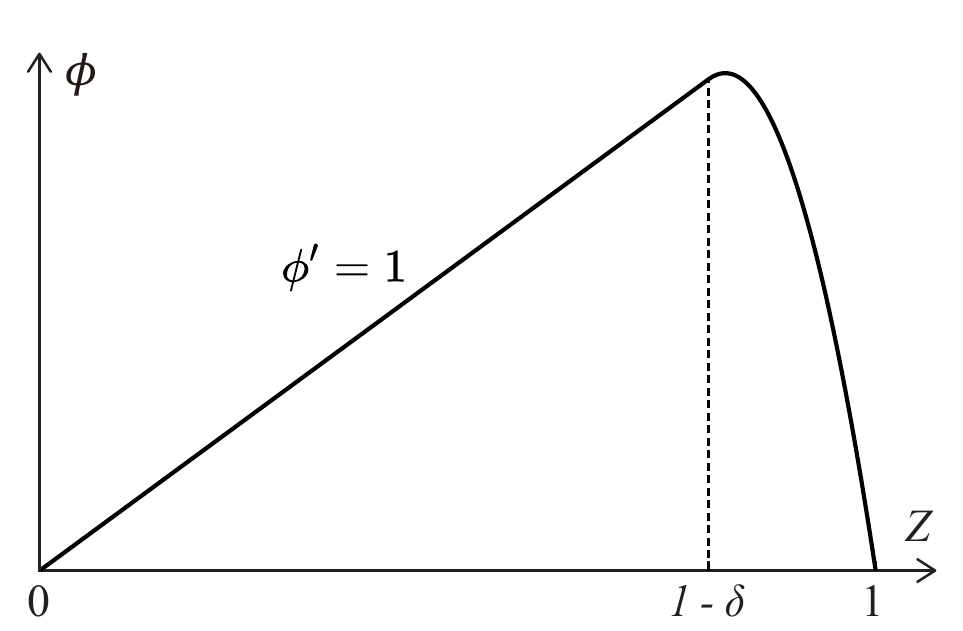}
    \caption{A sketch of the background field \eqref{eq_phi_wtheta} with a quadratic boundary layer and linear bulk.}
    \label{fig_quadraticBackground2}
\end{figure}

When the ansatz \eqref{eq_phi_wtheta} is substituted into \cref{eq_Program_wtheta_OF,eq_Program_wtheta_SC}, we get
\begin{equation}\label{eq_wtheta_B}
    B_{2}=\frac{\beta}{48}+\frac{1}{3\beta\delta}-\frac{1}{4\beta}+\frac{\delta}{6}-\frac{1}{4},
\end{equation}
and
\begin{equation}\label{eq_wtheta_Bk}
    Q_{\bm{k}}= \underbrace{ k^6\|{w}_{\bm{k}}\|_2^2+2\|{w}'_{\bm{k}}\|_2^2+\frac{1}{k^6}\|{w}''_{\bm{k}}\|_2^2 }_{P}-\frac{2{\mathcal{R}}}{\beta\delta^2}I_T,
\end{equation}
respectively. Therein, $I_T$ is the same as in \cref{eq_It}. 

\subsection{Sufficient conditions for the spectral constraint}\label{subsec_wtheta_SCest}

Similar to \cref{subsect_Theta_SCest}, the proof of \cref{Theorem_wtheta} relies on a lemma presented in this section, which provides a closed-form sufficient condition on the parameter of $\phi$ in \cref{eq_phi_wtheta} and on the balance parameter $\beta$ to ensure that the spectral constraint \eqref{eq_wtheta_SC} is satisfied. Before stating this result, we again present two auxiliary lemmas giving sufficient conditions for the non-negativity of $Q_{\bm{k}}$ over different ranges of wavenumbers.

\begin{lemma}[Sufficient conditions for $Q_{\bm{k}}\geq 0$, large $k$] \label{lemma_wtheta_largek}

Let $Q_{\bm{k}}$ be specified by \cref{eq_wtheta_Bk,eq_Theta_Bk_decompose}, and $k_c>0$ be a lower threshold wavenumber. Suppose that the positive parameter $\beta$ satisfies
\begin{equation}\label{eq_est_wtheta_SC_large}
    \frac{{\mathcal{R}} }{\beta}\leq k_c.
\end{equation}
Then, $Q_{\bm{k}}\geq0$ for all $k\geq k_c$.
\end{lemma}

\begin{proof}
    By applying Young's inequality along with inequalities \eqref{ineq_w'_inf} and \eqref{ineq_w_inf}, $P$ in \cref{eq_wtheta_Bk} is estimated from below as,
\begin{equation}
    \begin{aligned}\label{eq_est_wtheta_P}
        P&=k^6 \|w_{\bm{k}}\|_2^2+\|w_{\bm{k}}'\|_2^2+\|w_{\bm{k}}'\|_2^2+\frac{1}{k^6}\|w_{\bm{k}}''\|_2^2 \\
        &\geq 2k^3\|w_{\bm{k}}\|_2\|w_{\bm{k}}'\|_2+\frac{2}{k^3}\|w_{\bm{k}}'\|_2\|w_{\bm{k}}''\|_2\\
        &\geq k^3 \|w_{\bm{k}}\|^2_{\infty}+\frac{1}{k^3}\|w_{\bm{k}}'\|^2_{\infty}.
    \end{aligned}
\end{equation}
Substituting \eqref{eq_est_wtheta_P} and the upper bound for $I_T$ \eqref{eq_est_It_large} into \cref{eq_wtheta_Bk} yields the lower estimate on $Q_{\bm{k}}$,
\begin{equation}\label{eq_est_wtheta_Bk_large}
    Q_{\bm{k}}\geq \left(k-\frac{{\mathcal{R}}}{\beta}\right)\left(k^2 \|w_{\bm{k}}\|^2_{\infty}+\frac{1}{k^4}\|w_{\bm{k}}'\|_{\infty}^2\right).
\end{equation}
Under condition \eqref{eq_est_wtheta_SC_large}, the right-hand side of \eqref{eq_est_wtheta_Bk_large} is non-negative for all $k\geq k_c$.
\end{proof}

On the other hand, for smaller wavenumbers, the following lemma applies.

\begin{lemma}
    [Sufficient conditions for $Q_{\bm{k}}\geq 0$, small $k$] \label{lemma_wtheta_smallk}

Let $Q_{\bm{k}}$ be specified by \cref{eq_wtheta_Bk,eq_Theta_Bk_decompose}, $k_c>0$ be an upper threshold wavenumber, and $a\in(0,\,1)$ be an arbitrary number. Suppose that the positive pair $(\beta,\,\delta)$ satisfies
\begin{equation}\label{eq_est_wtheta_SC_small}
    \frac{{\mathcal{R}}}{\beta}\leq\min\left\{
    \frac{6a}{k_c^2 \delta},\quad 2\sqrt{\frac{5}{k_c \delta}}\left(\frac{2-2a}{3}\right)^{\frac{1}{4}}
    \right\}.
\end{equation}
Then, $Q_{\bm{k}}\geq0$ for all $0<k\leq k_c$.
\end{lemma}

\begin{proof}
    Applying the estimates \eqref{eq:IB_small} and \eqref{eq_est_P_small} to \cref{eq_wtheta_Bk} yields the lower bound,
\begin{equation}
\begin{aligned}
     Q_{\bm{k}}
     &\geq \left(2a-\frac{k^2\delta {\mathcal{R}} }{3\beta}\right)\|w_{\bm{k}}'\|_2^2+\frac{2}{k^{9/2}}\left[\left(\frac{2}{3}\right)^{\frac{3}{4}}(1-a)^{\frac{1}{4}}-\sqrt{\frac{k\delta}{30}}\frac{{\mathcal{R}} }{\beta}\right]\|w_{\bm{k}}'\|_{\infty}\|w_{\bm{k}}''\|_2.
\end{aligned}
\end{equation}
This lower estimate is non-negative provided that the pre-factors of both terms on the right-hand side are non-negative for all $k\in \left(0, \,k_c\right]$. This condition is equivalent to \eqref{eq_est_wtheta_SC_small}, which suffices for the non-negativity of $Q_{\bm{k}}$.
\end{proof}

Finally, taking the two auxiliary lemmas together, we arrive at a sufficient condition for the spectral constraint.

\begin{lemma}[Sufficient conditions for spectral constraint]
\label{lemma_SC_wtheta}
Let $\phi$ be defined by \cref{eq_phi_wtheta}. Suppose that the positive pair $(\beta,\,\delta)$ satisfies
\begin{equation}\label{eq_est_wtheta_SC}
    \frac{{\mathcal{R}} }{\beta}\leq c_{2}\, \delta^{-\frac{1}{3}},\quad\text{where}\quad c_{2}=\left[\frac{20}{9}\left(\sqrt{154}-10\right)\right]^{\frac{1}{3}}=1.750.
\end{equation}
Then, $(\beta,\,\delta)$ satisfies the spectral constraint.
\end{lemma}

\begin{proof} 
If we equate the right hand sides of \eqref{eq_est_wtheta_SC_large} and \eqref{eq_est_wtheta_SC_small}, then, solving the resulting two equations for $a$ and $k_c$ gives that,
\begin{equation}
    \label{eq:a_kc_wT}
    a=\frac{10}{27}\left(\sqrt{154}-10\right),\quad k_c=\left[\frac{20}{9}\left(\sqrt{154}-10\right)\right]^{\frac{1}{3}}\delta^{-\frac{1}{3}}.
\end{equation}
Given \eqref{eq:a_kc_wT}, both conditions of \cref{lemma_wtheta_largek,lemma_wtheta_smallk}, for large and small $k$ respectively, become equivalent to \eqref{eq_est_wtheta_SC}. Therefore, when \eqref{eq_est_wtheta_SC} holds, then $Q_{\bm{k}} \geq 0 $ for all $k>0$.
\end{proof}

    The velocity ansatz \eqref{eq_velocityansatz} may also be employed to demonstrate the optimality of the exponent $-\frac{1}{3}$ in the condition \eqref{eq_est_wtheta_SC}. Substituting \cref{eq_velocityansatz} into \cref{eq_wtheta_Bk}, applying the estimate \eqref{eq_Qk0ans_est} and rearranging, we obtain
    \begin{equation}
        \begin{aligned}
            Q_{k_a}^{ans} &= \frac{64}{105\delta}\left(\frac{1}{k_a^6\delta^2}+\frac{1}{12}+\frac{13}{4752}k_a^6\delta^2\right)-\frac{\mathcal{R}}{27\beta k_a\delta}\left(\frac{1}{k_a^3\delta}+\frac{19}{330}k_a^3\delta\right)\\
            &<\frac{64}{105\delta}\left(\frac{1}{k_a^3\delta}+\frac{19}{330}k_a^3\delta\right)\left(\frac{1}{k_a^3\delta}+\frac{65}{1368}k_a^3\delta-\frac{35\mathcal{R}}{576\beta k_a}\right).
        \end{aligned}
    \end{equation}
    With $k_a=\left(684/65\right)^{1/6}\delta^{-1/3}$ specified, a necessary condition for the spectral constraint to be satisfied is then 
    \begin{equation}
        \frac{{\mathcal{R}} }{\beta}<\frac{576}{35}\left(\frac{1}{k_a^2\delta}+\frac{65}{1368}k_a^4\delta\right)=\frac{144}{7}\left(
     \frac{78}{475}\right)^{\frac{1}{3}}\delta^{-\frac{1}{3}}.
    \end{equation}
    Therefore, no sufficient condition for the spectral constraint can take the form ${\mathcal{R}}/{\beta}\lesssim \delta^{-d}$ with $d>1/3$, if $\delta$ decreases as ${\mathcal{R}}$ increases beyond a certain threshold.


\subsection{Proof of theorem \ref{Theorem_wtheta}}

To complete the demonstration of \cref{Theorem_wtheta}, we define \textit{a posteriori}
\begin{equation}\label{eq_R0_wtheta}
    {\mathcal{R}}_{2} \coloneqq 2^{3}c_{2}=14.00,
\end{equation}
where $c_{2}$ is defined in \eqref{eq_est_wtheta_SC}, and choose in \cref{lemma_SC_wtheta}
\begin{equation}\label{eq_wtheta_choice}
    \beta=8\left(\frac{{\mathcal{R}}}{{\mathcal{R}}_{2}}\right)^{\frac{3}{5}},\quad \delta=\left(\frac{{\mathcal{R}}}{{\mathcal{R}}_{2}}\right)^{-\frac{6}{5}},
\end{equation}
which satisfy the condition \eqref{eq_est_wtheta_SC} and, based on \cref{proposition_boundingwtheta}, yield the upper bound on $\langle {w\theta} \rangle_{V,\,t}$,
\begin{equation}\label{eq_wtheta_bound}
    B_{2}=\frac{5}{24}\left(\frac{{\mathcal{R}}}{{\mathcal{R}}_{2}}\right)^{\frac{3}{5}}-\frac{1}{4}-\frac{1}{32}\left(\frac{{\mathcal{R}}}{{\mathcal{R}}_{2}}\right)^{-\frac{3}{5}}+\frac{1}{6}\left(\frac{{\mathcal{R}}}{{\mathcal{R}}_{2}}\right)^{-\frac{6}{5}}
    ,\quad \forall {\mathcal{R}}\geq {{\mathcal{R}}_{2}}.
\end{equation}
This bound is valid for ${\mathcal{R}}\geq{\mathcal{R}}_{2}$, where the requirement $\delta\leq 1$ is satisfied. For small $\mathcal{R}$, an unconstrained global minimum of \cref{eq_wtheta_B}, $0$, is achieved at $\delta = 1$ and $\beta = 2$. By the condition \eqref{eq_est_wtheta_SC}, this choice is feasible up to ${\mathcal{R}} = {\mathcal{R}}_{2}/4$. Since 0 is also the uniform lower bound, it follows that $\langle w\theta\rangle_{V,\,t}=0$ in this regime. For ${\mathcal{R}} > {\mathcal{R}}_2/4$, we then choose
\begin{equation}
    \delta=\sqrt{\frac{2}{\beta}}=\left(\frac{4{\mathcal{R}}}{{\mathcal{R}}_{2}}\right)^{-\frac{3}{7}},
\end{equation}
which gives the bound
\begin{equation}\label{eq_wtheta_bound2}
    B_{2}=\frac{1}{24}\left(\frac{4{\mathcal{R}}}{{\mathcal{R}}_{2}}\right)^{\frac{6}{7}}-\frac{1}{4}+\frac{1}{3}\left(\frac{4{\mathcal{R}}}{{\mathcal{R}}_{2}}\right)^{-\frac{3}{7}}-\frac{1}{8}\left(\frac{4{\mathcal{R}}}{{\mathcal{R}}_{2}}\right)^{-\frac{6}{7}}
    ,\quad\forall \mathcal{R}\geq \frac{1}{4}\mathcal{R}_{2}.
\end{equation}
This bound becomes suboptimal to \cref{eq_wtheta_bound} for ${\mathcal{R}} \geq 2.160\,{\mathcal{R}}_{2}$. With $\mathcal{R}_{2}$ specified by \eqref{eq_R0_wtheta}, \cref{Theorem_wtheta} is then proved.  \hfill $\square$

Eq. \eqref{eq_wtheta_bound} is suboptimal once it exceeds the uniform upper bound $\langle {w\theta}\rangle_{V,\,t}\leq 1/2$, which occurs at ${\mathcal{R}}=8.376\,{\mathcal{R}}_{2}=117.2$. The upper bounds on $\langle w\theta\rangle_{V,\,t}$ for different ranges of ${\mathcal{R}}$ are displayed in figure \ref{fig_wtheta_bounds}, where the analytical bounds lie at most 10\% above the numerical solution of minimising \cref{eq_wtheta_B} subject to the constraint \eqref{eq_est_wtheta_SC} and $\delta\leq1$.

\begin{figure}[h]
    \centering
    \includegraphics[width=0.8\linewidth]{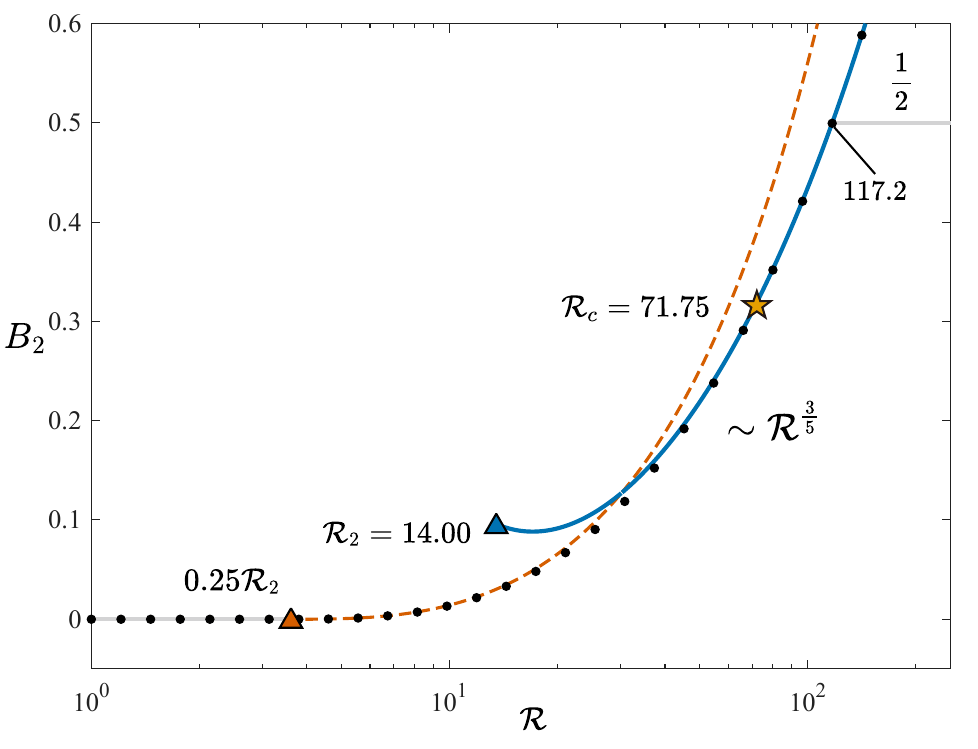}
    \caption{Upper bounds on $\langle w\theta \rangle_{V,\,t}$ from \eqref{eq_wtheta_bound} (blue solid line) and \eqref{eq_wtheta_bound2} (red dashed line), compared with the uniform bounds of $0$ and $\frac{1}{2}$ (grey solid lines), along with the numerical minimum (black dots) of \cref{eq_wtheta_B} subject to the condition \eqref{eq_est_wtheta_SC} and $\delta \leq 1$. The triangle markers denote the ${\mathcal{R}}$ above which \cref{eq_wtheta_bound,eq_wtheta_bound2} are valid, and the star indicates the linear onset of convection.}
    \label{fig_wtheta_bounds}
\end{figure}

\section{Numerically optimised bounds}
\label{sec:numerical}
\begin{figure}[h]
    \centering
    \includegraphics[width=\textwidth]{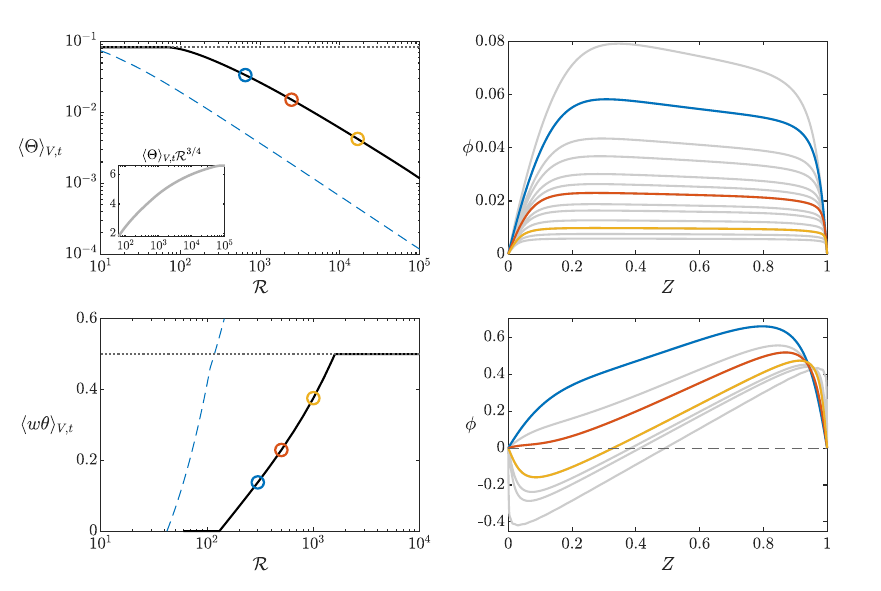}
    \caption{Numerically optimised bounds and corresponding background fields $\langle\Theta\rangle_{V,\,t}$ (top row) and $\langle w\theta\rangle_{V,\,t}$ (bottom row). The left panels show bounds as functions of the reduced Rayleigh number $\mathcal{R}$, while the right panels show the associated optimised background fields $\phi(Z)$. Top: numerical lower bound on $\langle\Theta\rangle_{V,\,t}$ (black solid line), the best applicable analytical bound from \eqref{eq_bound_Theta} and \eqref{eq_bound_Theta_2} (blue dashed line), and the uniform upper bound $1/12$ (horizontal dotted line). The inset shows the compensated numerical bound $\langle\Theta\rangle_{V,\,t}\mathcal{R}^{3/4}$. The highlighted background fields, in blue, red and yellow, correspond respectively to $\mathcal{R}=6.51\times10^{2}$, $2.47\times10^{3}$ and $8.15\times10^{3}$. Bottom: numerical upper bound on $\langle w\theta\rangle_{V,\,t}$ (black solid line), the uniform bound $1/2$ (horizontal dotted line), and the best applicable analytical bound from \eqref{eq_wtheta_bound} and \eqref{eq_wtheta_bound2} (blue dashed line). The blue, red and yellow background fields correspond respectively to $\mathcal{R}=300$, $500$ and $1000$. Grey curves show a selection of the remaining computed background fields.}
    \label{fig:numerical_results}
\end{figure}

To assess the sharpness of the analytical results, we numerically approximate the optimal bounds of the variational problems in \cref{proposition_boundingTheta} and \cref{proposition_boundingwtheta}. 
Following a Legendre discretisation used in related background-method computations \citep{pachev2020rigorous,fantuzzi2022background}, the background gradient $\phi'$ and the vertical velocity $w$ are represented by truncated Legendre expansions. This discretisation reduces the spectral constraint at each horizontal wavenumber to a positive semidefiniteness condition. The resulting finite dimensional conic programmes are implemented using YALMIP \citep{Lofberg2004}, which is a MATLAB-based optimisation modelling framework, and solved with MOSEK \citep{MOSEK2026}, which is an optimisation solver. The nominal polynomial degrees for $\phi'$ and $w$ are 50 and 60, respectively, selected cases are recomputed at degrees as high as 80 and 90 until the bound changes by less than 1\% for $\langle \Theta \rangle_{V,t}$ up to $\mathcal{R}=10^5$ and 0.02\% for $\langle w\theta \rangle_{V,t} < 1/2$.
The wavenumber grid is iteratively refined near the minima of the smallest spectral constraint eigenvalue.
The resulting bounds and optimised background fields are shown in \cref{fig:numerical_results} on the left- and right-hand panels respectively.

\section{Discussion}
\label{Discussion}

In this work, we present the asymptotically reduced equations for rapidly rotating internally heated convection and prove that, in the limit of infinite $Pr$, under stress-free and fixed temperature boundary conditions, the mean temperature of the system is bounded from below by $\langle\Theta\rangle_{V,\,t}\gtrsim R^{-3/4}E^{-1}$, while the mean heat transport is bounded from above by $\langle{w\theta}\rangle_{V,\,t} \lesssim R^{3/5}E^{4/5}$, for $RE^{4/3}$ larger than a given threshold. We demonstrate that for the specific class of background profiles selected the bounds scale optimally within the methodology and carry out a numerical optimisation to determine the optimal bounds and background profiles.

\begin{figure}
    \centering
    \includegraphics[width=\linewidth]{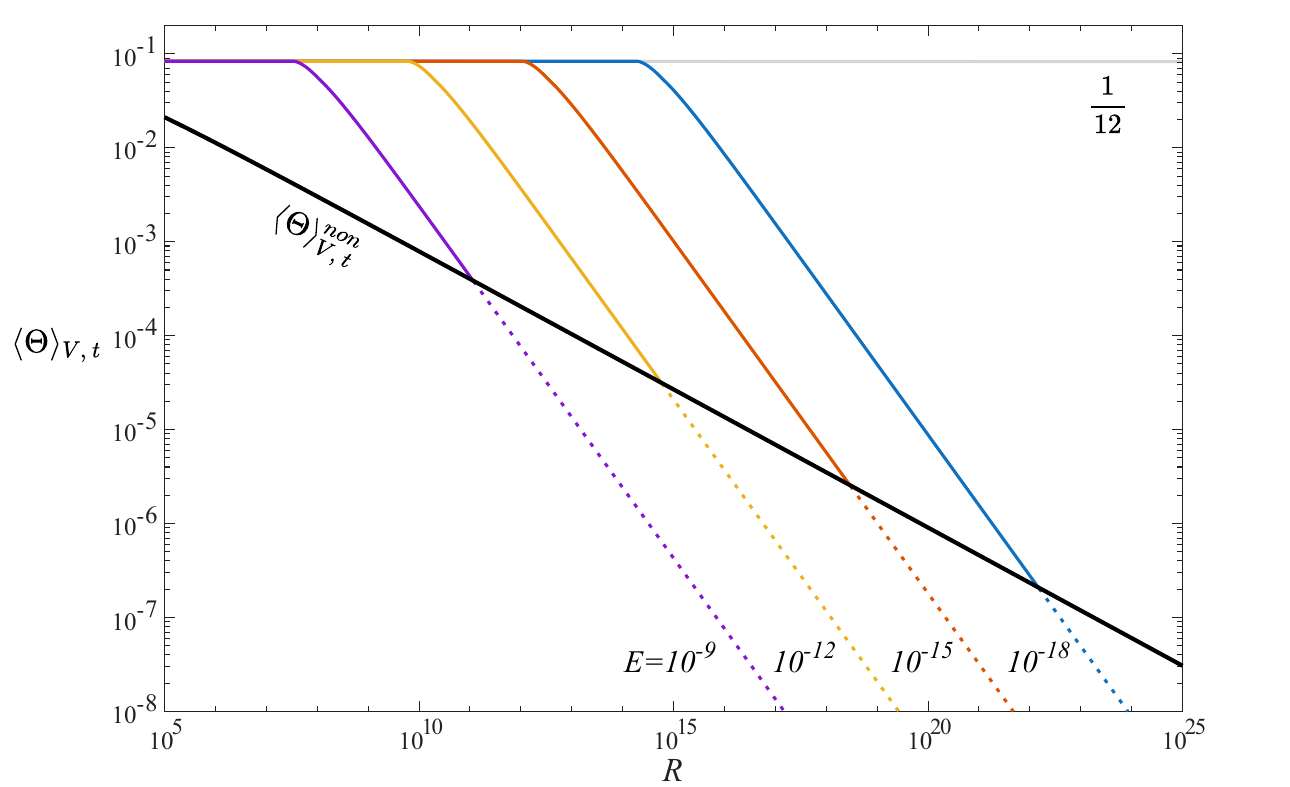}
    \caption{Lower bounds on $\langle\Theta\rangle_{V,\,t}$ from \cref{eq_bound_Theta,eq_bound_Theta_2} as functions of $R$ for different values of $E$ (coloured solid lines). The black solid line shows the bound for non-rotating convection, \cref{eq_bound_Theta_nonrot}. Dotted lines indicate regions where the bounds no longer apply.}
    \label{fig_bounds}
\end{figure}

Figure \ref{fig_bounds} compares lower bounds of $\langle\Theta\rangle_{V,\,t}$ from \cref{eq_bound_Theta,eq_bound_Theta_2} to the non-rotating lower bound from \citet{whitehead2012rigid} of,
\begin{equation}\label{eq_bound_Theta_nonrot}
    \langle\Theta\rangle^{non}_{V,\,t}\geq 0.6910 R^{-\frac{5}{17}}\left(1-2.849R^{-\frac{5}{17}}\right).
\end{equation}
At fixed $E$, the bound on $\langle\Theta\rangle_{V,\,t}$ 
scales with a smaller power of $R$ than $\langle\Theta\rangle_{V,\,t}^{non}$. 
Although the bound scales with a smaller power of $R$ under rotation, the absolute magnitude of $\langle\Theta\rangle_{V,\,t}$ remains larger than $\langle\Theta\rangle^{non}_{V,\,t}$ while the flow is rotationally constrained.
That means \cref{Theorem_Theta} provides better bounds only for $R\lesssim E^{-68/31}$. For larger $R$, however, the least stringent criterion reported for rotation dominance \citep{kunnen2021geostrophic}, $R\lesssim E^{-2}$, is not satisfied, and the flow lies outside of the QG regimes of interest. For the other quantity of interest, lower bounds on $\mathcal{F_B}$ can be obtained from \cref{Theorem_wtheta}, by \cref{eq_FB}, and more heat tends to exit from the lower boundary as rotation becomes more rapid. Although a non-rotating bound $\mathcal{F}_B\geq 1/2-2^{-21/5}R^{1/5}$ has been demonstrated for no-slip boundary conditions \citep{arslan2021bounds}, its validity is restricted to $R\leq 65536$, hence a comparison cannot be made.

The numerical results in \cref{sec:numerical} show that more general background fields yield improvements over the explicit analytical constructions throughout the range of $\mathcal{R}$ considered. For $\langle\Theta\rangle_{V,\,t}$, the numerically optimised lower bound remains above the analytical bound. The corresponding optimal background fields develop an increasingly flat bulk as $\mathcal{R}$ increases, approaching the ansatz (\cref{fig_quadraticBackground}), but retain unequal boundary layers with a thinner and steeper upper boundary layer. The numerical bound decays more slowly than the analytical $\mathcal{R}^{-3/4}$ scaling across this range.
For $\langle w\theta \rangle_{V,t}$, the numerical upper bound remains below the uniform bound $1/2$ up to $\mathcal{R} \sim 1.59\times10^{3}$, whereas the explicit analytical bound becomes larger than $1/2$ at $\mathcal{R}=117.2$. The optimised background fields contain unequal boundary layers at both boundaries, in contrast to the single upper boundary layer of the ansatz in \cref{fig_quadraticBackground2}, while their bulk gradient approaches unity as $\mathcal{R}$ increases. These results suggest that sharper analytical bounds at moderate $\mathcal{R}$ would require an asymmetric $\phi(Z)$ with two independently parametrised boundary layers, and a gradient in the bulk tending to 0 or 1, for the two quantities respectively, as $\mathcal{R}$ increases. Proving an improvement with the aforementioned features in $\phi$ remains an open analytical problem. Since solutions of the governing equations constitute a subset of the admissible fields in the optimisation, it remains unclear whether the optimal bounds can be saturated by a physical flow. Moreover, the differences between the optimisers for the mean temperature and the mean heat flux do not preclude the existence of a physical flow that simultaneously attains the optimal values of both quantities over solutions of the governing equations.

A diagnostic Rayleigh number, $R_{diag}\coloneqq R/Nu_p$ with $Nu_p\sim \langle\Theta\rangle_{V,\,t}^{-1}$ being a proxy Nusselt number quantifying the intensity of mixing \citep{goluskin2016internally}, can be introduced to facilitate comparison between bounds on $\langle \Theta \rangle_{V,\,t}$ in internally heated convection and those on the Nusselt number $Nu$, quantifying the convective enhancement of heat transport, in Rayleigh-B\'enard convection. With this definition, \cref{Theorem_Theta} can be written as $Nu_p \lesssim R_{diag}^{3}$ for fixed $E$, which coincides with the upper bound on $Nu$ for NH-QG RBC with stress-free and fixed-temperature boundary conditions \citep{grooms2014bounds}. The equivalence between our bounds on $\langle \Theta \rangle_{V,\,t} $ and those on $Nu$ arises from the fact that the spectral constraint in both variational problems (equation 18 in \citet{grooms2014bounds} and \cref{eq_Theta_SC,eq_Program_Theta_SC}) is structurally identical. By using the Green’s function solution of the diagnostic equation \eqref{eq_ODE} to estimate $w_{\bm{k}}$ in terms of $\theta_{\bm{k}}$, and assuming \textit{a priori} a threshold wavenumber $k_c$ scaling as $\delta^{-1/3}$, \citet{grooms2014bounds} establish a sufficient condition with the optimal scaling ${\mathcal{R}}\lesssim \delta^{-1/3}$. This is equivalent to \eqref{eq_est_SC} upon fixing $A=1/2$, despite differences in the definition of $\phi$. In contrast to the Green’s function approach, we invoke \eqref{eq_ODE} to eliminate $\theta_{\bm{k}}$, demonstrate conditions for the validity of the spectral constraint at large and small wavenumbers with integral inequalities, and
solve for $k_c$ that delineates the two regions. Solving for $k_c$ provides a single sufficient condition for the validity of the spectral constraint, and this condition has the same scaling as that Grooms and Whitehead find in their problem.

The threshold wavenumber $k_c$ identifies the flow length scale at which the spectral constraint is most restrictive. The scaling $k_c\sim\delta^{-1/3}$ indicates that the analytical bound is determined by progressively smaller structures as $\mathcal{R}$ increases. At the analytical optimum of $\langle \Theta \rangle_{V,\,t}$, we further find that $k_c\sim \mathcal{R}^{1/4}$, corresponding to a dimensional horizontal length scale $l_c\sim R^{-1/4 }H$, independent of $E$. It is unclear whether this length scale encodes a deeper physical meaning. 

The upper bound for $\langle {w\theta} \rangle_{V,\,t}$ given by \cref{Theorem_wtheta} exceeds the uniform upper bound $1/2$ at large ${\mathcal{R}}$. This has been demonstrated to correspond to a violation of the minimum principle $T\geq0$ in previous literature \citep{arslan2021bounds}. Previous studies of IHC address this issue by imposing the minimum principle as an additional constraint in \eqref{eq_supwtheta} \citep{arslan2021bounds,arslan2023rigorous}, which changes the expression of the bound \eqref{eq_Program_wtheta_OF}, without affecting the spectral constraint. Despite pursuing a similar strategy, we were unable to establish a bound that tends to $1/2$ from below, using background profiles featuring quadratic, linear or multiple sub-layered boundary layers. The difficulty arises from the $\beta$-dependence of the bound $B$ when $\delta$ is constrained by \eqref{eq_est_wtheta_SC}, where $\beta$ now varies with $\mathcal{R}$ as opposed to being a constant, as is the case for bounding $\langle \Theta\rangle_{V,\,t}$. 
That being said, the governing equations, \eqref{eq_QG}, are a proper reduction of the full convection system formulated in \cref{app_A1} only for $\mathcal{R}$ below a finite threshold. When we fix $E$, if $R$ is sufficiently large, convection is dominated by buoyancy as opposed to rotation, and \eqref{eq_QG} no longer captures the dominant force balance in the full dynamics. Therefore, the bound on $\langle w \theta \rangle_{V,\,t}$ is relevant for the geostrophic regime valid in the range $0\leq\langle w \theta \rangle_{V,\,t} \leq C$, where $C<1/2$ and the asymptotic limit of $1/2$ is not relevant to our study but is instead addressed by bounds on IHC without rotation. The value of the aforementioned $C=C(E,R)$ remains an open question, and the parameter regime over which the reduced equations are valid could be investigated in future work by comparison with solutions of the full PDE system.

The bounds we prove are independent of the horizontal periods of the fluid domain. Since the spectral constraint is required to be satisfied for every positive wavenumber. Unlike the non-rotating bounds in \citet{lu2004bounds}, our methodology for proving \cref{Theorem_Theta,Theorem_wtheta} cannot be directly extended to general $Pr$. Without the diagnostic relation \eqref{eq_ODE} arising in the infinite-$Pr$ limit, or supplying additional information, the non-negativity of $Q_{\bm{k}}$ cannot be ensured, even when the kinetic energy balance is taken into account. Nevertheless, future work may follow the approaches of \citet{wang2013bound}, where the solutions at infinite $Pr$ are treated as perturbations from the solutions when $Pr<\infty$, and 
estimates on inertia terms lead to corrected bounds applicable to large but finite $Pr$. In \citet{tilgner2022bounds}, a similar approach is taken to obtain a semi-analytical bound on rotating convection, albeit by only using the maximum principle from the temperature equation and not any dynamical information. Another approach would be to use different auxiliary functionals, $\mathcal{V}$, although higher-order polynomials remain analytically intractable such that non-polynomial alternatives may be required \citep{Goluskin_2019,CHERNYSHENKO2023}. Finally, although comparison of the theoretical bounds with simulations or experiments would be valuable, no such data are currently available for rotation-dominated IHC with uniform heating. Existing studies have primarily focused on the onset of convection \citep{sanchez2016critical}, subcritical behaviour \citep{kaplan2017subcritical}, regime transitions \citep{lin2021large}, and characteristic scales \citep{guervilly2019turbulent}, while systematic investigations of global diagnostics remain lacking. The closest available results are those of \citet{ostilla2025rotationally}, but these are restricted to the rotation-affected regimes. A systematic numerical investigation of the asymptotically reduced model is left for future work.

\begin{bmhead}[Acknowledgements]
We sincerely thank the three reviewers for their constructive suggestions, which have helped improve the clarity and presentation of this article.
\end{bmhead}

\begin{bmhead}[Funding]
This study was funded by the Swiss National Science Foundation Grant No. 227702 (10000683), for which we are grateful.
We also acknowledge funding by the European Research Council (agreement No. 833848-UEMHP) under the Horizon 2020 program.
\end{bmhead}

\begin{bmhead}[Declaration of interests]
The authors report no conflict of interest.
\end{bmhead}

\begin{bmhead}[Author ORCIDs]
Y. Zhang, https://orcid.org/0000-0003-4811-2055; A. Arslan, https://orcid.org/0000-0002-5824-5604; S. Maffei, https://orcid.org/0000-0002-1603-4040; A.Jackson, https://orcid.org/0000-0003-1821-4114.
\end{bmhead}

\begin{appen}
\section{Reduced model for rapidly rotating internally heated convection}
\label{app:eqs}
\label{A_1_demonstration}

This appendix outlines the derivation of the reduced model for rapidly rotating internally heated convection, \eqref{eq_QG}, based on the setup in \cref{sec_setup}. The derivation follows the same methodology as that used to obtain the Rayleigh–Bénard NH-QG equations in \citet{julien1998new} and \citet{sprague2006numerical}, and begins with the standard dimensional Navier–Stokes, temperature, and continuity equations for Boussinesq rotating convection with uniform internal heating, in the Cartesian coordinate system $(x,y,z)$, 
\begin{subequations}\label{eq_NS}
    \begin{gather}
        \partial_t \bm{u}+\bm{\nabla\cdot}\left(\bm{u}\bm{u}^T\right)+2\Omega\hat{\bm{z}}\times\bm{u}=-\frac{1}{\rho}\nabla p+\alpha gT \hat{\bm{z}}+\nu\nabla^2\bm{u},\\
        \partial_tT+\bm{\nabla\cdot}\left(\bm{u}T\right)=\kappa\nabla^2T+\frac{S}{\rho c_p},\\
        \bm{\nabla\cdot}\bm{u}=0. \label{eq_NS_continuity}
    \end{gather}
\end{subequations}
Here, $\bm{u}=\bm{u}(x,y,z,t)$, $T=T (x,y,z,t)$, and $p=p(x,y,z,t)$ denote the dimensional velocity, temperature and pressure fields. For an incompressible fluid, conservation of mass reduces to the divergence-free condition \eqref{eq_NS_continuity}.

For the reduced equations to be applicable to general $Pr$, we derive them on viscous time scales and then, to obtain \cref{eq_QG}, rescale to thermal time scales and take the infinite $Pr$ limit, following sections 2.4 and 2.5 of \citet{sprague2006numerical}. Motivated by the inherent anisotropy between flow scales along the rotation axis and those in the horizontal planes in QG regimes, an anisotropic non-dimensionalisation is adopted. Specifically, the layer height $H$ is taken as the vertical length scale, while the horizontal characteristic flow scale $L$ is used as the horizontal length scale, with
\begin{equation}
\epsilon \coloneqq \frac{L}{H} \ll 1.
\end{equation}
Moreover, to reflect the separation between fast horizontal dynamics and the slow vertical evolution subject to rapid rotation, the time scale for the fluctuating fields about the horizontal mean is chosen as the horizontal viscous diffusion time $L^{2}/\nu$, while the horizontally averaged fields evolve on the slower vertical diffusion time $H^{2}/\nu$. Aside from these, the velocity and temperature units are isotropically chosen as the horizontal viscous diffusion velocity $\nu/L$ and the conductive temperature $S H^{2}/\rho c_p \kappa$, respectively. The pressure scale is fixed as $\rho (\nu/L)^{2} (H/L)^{2}$. With these choices, the dimensionless form of \cref{eq_NS} after decomposing the velocity, temperature and pressure into a horizontal mean and fluctuations, is
\begin{subequations}\label{eq_NS_dimensionless}
    \begin{gather}
        \partial_t \tilde{\bm{u}}+\epsilon^2 \partial_{t_s} \,\overline{\bm{u}}+\left(\nabla_h+\epsilon\hat{\bm{Z}}\partial_Z\right)\cdot\left[\left(\overline{\bm{u}}+\tilde{\bm{u}}\right)\left(\overline{\bm{u}}+\tilde{\bm{u}}\right)^T\right]+\frac{\epsilon^2}{E}\hat{\bm{Z}}\times\left(\overline{\bm{u}}+\tilde{\bm{u}}\right)\nonumber\\
        =-\left( \epsilon^{-2}\,\nabla_h+\epsilon^{-1} \hat{\bm{Z}}\partial_Z\right)p+\frac{\epsilon^3 R}{Pr}\left(\overline{T}+\tilde{T}\right) \hat{\bm{Z}}+\left(\nabla_h^2+\epsilon^2 \partial_Z^2\right)\left(\overline{\bm{u}}+\tilde{\bm{u}}\right),\\
        \partial_t \tilde{T}+\epsilon^2\partial_{t_s} \,\overline{T}+\left(\nabla_h+\epsilon\hat{\bm{Z}}\partial_Z\right)\cdot\left[\left(\overline{\bm{u}}+\tilde{\bm{u}}\right)\left(\overline{T}+\tilde{T}\right)\right]=\frac{1}{Pr}\left(\nabla_h^2+\epsilon^2 \partial_Z^2\right)\left(\overline{T}+\tilde{T}\right)+\frac{\epsilon^2}{Pr},\\
       \left(\nabla_h+\epsilon\hat{\bm{Z}}\partial_Z\right)\cdot\left(\overline{\bm{u}}+\tilde{\bm{u}}\right)=0.
    \end{gather}
\end{subequations}
Therein, $t_s\coloneqq \epsilon^2 t$, the capitalised $Z$ denotes the dimensionless vertical coordinate to highlight the anisotropy, and $\tilde{\left(\cdot\right)}\coloneqq \left(\cdot\right)-\overline{\left(\cdot\right)}$ denote the fluctuation of a field about its horizontal mean. 

In \citet{sprague2006numerical}, a small-scale dimensionless vertical coordinate $z$ is retained to capture vertical variations on spatial scales comparable to $L$, which are then removed by averaging over those scales. In contrast, the present formulation assumes the absence of vertical modes on such scales, following \citet{julien1998new}. This assumption is appropriate for quasi-geostrophic regimes, where rapid vertical variations are inhibited by the dominant geostrophic balance. We then state below the asymptotically reduced system from \cref{eq_NS_dimensionless} in the limit of rapid rotation, along with some necessary definitions . 

Let $\bm{u}$, $p$, and $T$ be solutions of the dimensionless Boussinesq system for internally heated rotating convection \eqref{eq_NS_dimensionless}, satisfying impermeable, stress-free boundary conditions at $Z=0,1$ and horizontal periodicity. Decompose the velocity field as $\bm{u}=\bm{u}_h+w\hat{\bm{Z}}$, with $\bm{u}_h(x,y,Z,t)$ and $w(x,y,Z,t)$ being the horizontal and the vertical components, respectively, and define the vertical vorticity by $\zeta\coloneqq \hat{\bm{Z}} \cdot \left(\nabla_h\times {\bm{u}}_{h}\right)$. In the limit of rapid rotation, suppose that $\epsilon=E^{\frac{1}{3}}\to0$, with the reduced Rayleigh number ${\mathcal{R}} \coloneqq R \epsilon^4$ remaining $O(1)$. Further suppose that the dimensionless fields $\bm{u}$, $p$, $T$ and $\zeta$ admit the asymptotic expansion
\begin{equation}\label{eq_asymptoticexpansion}
    f=\sum_{n=0}^{\infty}\epsilon^nf_n=f_0+\epsilon f_1+\epsilon^2f_2+O(\epsilon^3),
\end{equation}
in which $n\in\mathbb{N}$, and $f_n=f_n (x,y,Z,t)$ is independent of $\epsilon$. Then, the dynamics of the leading order fields follow the dimensionless PDEs,
\begin{subequations}\label{eq_QG_viscousscale}
\begin{gather}
        \partial_t \tilde{w}_0+J[\tilde{p}_1,\tilde{w}_0]+\partial_Z\tilde{p}_1=\frac{{\mathcal{R}}}{Pr}\tilde{T}_1 +\nabla_h^2 \tilde{w}_0,\label{eq_w_originalscale}\\
    \partial_t \tilde{\zeta}_0+J\left[\tilde{p}_1,\,\tilde{\zeta}_0\right]-\partial_Z\tilde{w}_0=\nabla_{h}^2 \tilde{\zeta}_0,\label{eq_zeta_originalscale}\\
    \tilde{\zeta}_0=\nabla_{h}^2\,\tilde{p}_1,\label{eq_p_originalscale}\\
    \partial_t\tilde{T}_1+J[\tilde{p}_1,\,\tilde{T}_1]+\tilde{w}_0\partial_Z \overline{T}_0=\frac{1}{Pr}\nabla_{h}^2 \tilde{T}_1 \label{eq_theta_originalscale}\\
    E^{-\frac{2}{3}}\partial_t\overline{T}_0+\partial_Z\left(\overline{\tilde{w}_0\tilde{T}_1}\right)=\frac{1}{Pr}\partial_Z^2\overline{T}_0+\frac{1}{Pr}.\label{eq_Theta_originalscale}
\end{gather}    
\end{subequations}

 As will be shown in \cref{app_A1}, $\tilde{p}_1$ is also a stream function for the leading-order horizontal velocity field, with $\tilde{\bm u}_{h\,0}=-\hat{\bm x}\partial_y\tilde{p}_1+\hat{\bm y}\partial_x\tilde{p}_1$. From \cref{eq_QG_viscousscale}, the reduced model in the infinite-$Pr$ limit, given by \cref{eq_QG}, can be derived upon applying the rescaling 
    \begin{equation}
    t\to Pr\,t,\quad (w,\,\zeta,\,\,p)\to\frac{1}{Pr}(w,\,\zeta,\,\,p),
    \end{equation}
    and relabelling $\tilde{w}_0$, $\tilde{\zeta}_0$, $\tilde{p}_1$, $\tilde{T}_1$ and $\overline{T}_0$ as $w$,  $\zeta$, $\psi$, $\theta$ and $\Theta$ respectively, in order to match the notation used in the main text. However, throughout the remainder of this appendix we retain the original variables without relabelling or rescaling.

\subsection{Derivation of the reduced equations}
\label{app_A1}
Now we derive \cref{eq_QG_viscousscale}, which begins by taking the horizontal average of \cref{eq_NS_dimensionless} to obtain the mean-field dynamics,
    \begin{subequations}\label{eq_NS_mean}
    \begin{gather}
        \epsilon^{-1}\left[\hat{\bm{Z}}\times \overline{\bm{u}}_h+\hat{\bm{Z}}\left(\partial_Z \,\overline{p}-\frac{{\mathcal{R}}}{Pr}\overline{T}\right)
        \right]+\epsilon \, \partial_Z \left(\overline{w\bm{u}}\right)+\epsilon^2\left(\partial_{t_s}\, \overline{\bm{u}}-\partial_Z^2\,\overline{\bm{u}}\right)=0,\label{eq_NS_mean_u}\\
        \epsilon\,\partial_Z \left(\overline{wT}\right)+\epsilon^2\left(\partial_{t_s} \,\overline{T}-\frac{1}{Pr}\partial_Z^2 \,\overline{T}-\frac{1}{Pr}\right)=0,\label{eq_NS_mean_Theta}\\      \epsilon\,\partial_Z\overline{w}=0.\label{eq_NS_mean_w}
    \end{gather}
\end{subequations}
Therein, terms are organized in ascending powers of $\epsilon$ to highlight the leading-order balances. From \cref{eq_NS_mean_w} and the impermeable boundary conditions, it directly follows that 
\begin{equation}\label{eq_mean_w}
    \overline{w}=0.
\end{equation}
Subtracting \cref{eq_NS_mean} from \cref{eq_NS_dimensionless} and using \cref{eq_mean_w} further yield the evolution of the fluctuating fields, with terms collected by ascending powers of $\epsilon$,
\begin{subequations}\label{eq_NS_fluc}
\begin{gather}
    \epsilon^{-2}  \nabla_h\,\tilde{p}+\epsilon^{-1}\left[\hat{\bm{Z}}\times\tilde{\bm{u}}_h+\hat{\bm{Z}}\left(\partial_Z \,\tilde{p}-\frac{{\mathcal{R}}}{Pr}\tilde{T}\right)\right]+\left[\partial_t+\left(\overline{\bm{u}}_h+\tilde{\bm{u}}_h\right)\cdot\nabla_h-\nabla_h^2\right]\tilde{\bm{u}}\nonumber\\
    +\epsilon\left[\tilde{w}\partial_Z \left(\overline{\bm{u}}+\tilde{\bm{u}}\right)-\partial_Z \left(\overline{\tilde{w}\tilde{\bm{u}}}\right)\right]-\epsilon^2  \partial_Z^2 \tilde{\bm{u}}=0,
    \label{eq_NS_fluc_u}\\
    \left[\partial_t+\left(\overline{\bm{u}}_h+\tilde{\bm{u}}_h\right)\cdot\nabla_{h} -\frac{1}{Pr}\nabla_h^2\right]\tilde{T}+\epsilon\left[\tilde{w}\partial_Z \left(\overline{T}+\tilde{T}\right)-\partial_Z  \left(\overline{\tilde{w}\tilde{T}}\right)\right]-\frac{\epsilon^2}{Pr} \partial_Z^2 \tilde{T}=0, \label{eq_NS_fluc_Theta}\\
    \nabla_{h}\cdot\tilde{\bm{u}}_{h}+\epsilon\,\partial_Z\tilde{w}=0.\label{eq_NS_fluc_w}
\end{gather}
\end{subequations}

The next step is to substitute \cref{eq_asymptoticexpansion} into \cref{eq_NS_mean,eq_NS_fluc}, and compare terms at successive orders of $\epsilon$, which yields the force and flux balances at each order. Applying this to \cref{eq_NS_mean_u}, from the $O(\epsilon^{-1} )$ and $O(1)$ terms we have
\begin{equation}
    \hat{\bm{Z}}\times \overline{\bm{u}}_{h \,n}+\hat{\bm{Z}}\left(\partial_Z \,\overline{p}_{n} -\frac{{\mathcal{R}}}{Pr}\overline{T}_{n} \right)=0,\quad n=0,\,1,
        \label{eq_mean_u_minus1}
\end{equation}
the horizontal component of which implies
\begin{equation}\label{eq_mean_uh}
    \overline{\bm{u}}_{h\,n}=0,\quad n=0,\,1.
\end{equation}

With \cref{eq_mean_w} utilized, $O(\epsilon)$ and $O(\epsilon^2 )$ terms of the mean temperature equation \eqref{eq_NS_mean_Theta} yield
\begin{subequations}
    \begin{gather}
        \partial_Z \left(\overline{\tilde{w}_0\tilde{T}_0}\right)=0,\label{eq_mean_Theta_1}\\
        \partial_Z\left(\overline{\tilde{w}_0\tilde{T}_1+\tilde{w}_1\tilde{T}_0}\right)+\partial_{t_s} \,\overline{T}_0-\frac{1}{Pr}\partial_Z^2 \,\overline{T}_0 -\frac{1}{Pr} =0;\label{eq_mean_Theta_2}
    \end{gather}
\end{subequations}
while, further aided by \cref{eq_mean_uh}, the $O(1)$ and $O(\epsilon )$ terms of the temperature fluctuation equation \eqref{eq_NS_fluc_Theta} lead to
\begin{subequations}
    \begin{gather}
        \left(\partial_t+\tilde{\bm{u}}_{h\, 0}\cdot\nabla_{h}-\frac{1}{Pr}\nabla_{h}^2\right) \,\tilde{T}_0 = 0,\label{eq_fluc_Theta_0}\\
         \left(\partial_t+\tilde{\bm{u}}_{h \, 0}\cdot\nabla_{h}-\frac{1}{Pr}\nabla_h^2\right)\,\tilde{T}_1 +\tilde{\bm{u}}_{h\, 1}\cdot\nabla_{h}\tilde{T}_0+\tilde{w}_0\partial_Z\left(\overline{T}_0+\tilde{T}_0\right)=0.\label{eq_fluc_Theta_1}
    \end{gather}
\end{subequations}
Here, the relation \eqref{eq_mean_Theta_1} has been used to simplify \cref{eq_fluc_Theta_1}.

The $O(\epsilon^{-2})$, $O(\epsilon^{-1})$ and $O(1)$ force balances in the momentum fluctuation equation \eqref{eq_NS_fluc_u} are,
\begin{subequations}
    \begin{gather}
    \nabla_h \,\tilde{p}_0=0,\label{eq_fluc_u_minus2}\\   
        \nabla_h \,\tilde{p}_1 + \hat{\bm{Z}}\times \tilde{\bm{u}}_{h\, 0}+\hat{\bm{Z}}\left(\partial_Z \,\tilde{p}_0-\frac{{\mathcal{R}}}{Pr} \tilde{T}_0\right)=0,
        \label{eq_fluc_u_minus1}\\
       \nabla_h \,\tilde{p}_2 + \hat{\bm{Z}}\times \tilde{\bm{u}}_{h\, 1}+\hat{\bm{Z}}\left(\partial_Z \,\tilde{p}_1-\frac{{\mathcal{R}}}{Pr} \tilde{T}_1\right)+\left(\partial_t+\tilde{\bm{u}}_{h\, 0}\cdot\nabla_{h}-\nabla_h^2 \right)\tilde{\bm{u}}_0 =0,\label{eq_fluc_u_0}
    \end{gather}
\end{subequations}
to which \cref{eq_mean_uh} has been inserted. By the definition of the fluctuating fields,  \cref{eq_fluc_u_minus2} indicates
\begin{equation}
    \tilde{p}_0=0.
\end{equation}
Substituting this result into the vertical component of \cref{eq_fluc_u_minus1} then yields
\begin{equation}\label{eq_theta0}
    \tilde{T}_0=0,
\end{equation}
consistent with \cref{eq_mean_Theta_1,eq_fluc_Theta_0}. Moreover, the horizontal component of \cref{eq_fluc_u_minus1} gives 
\begin{equation}\label{eq_geostrophy}
    \bm{e}_Z \times\tilde{\bm{u}}_{h\, 0}=-\nabla_{h}\,\tilde{p}_1,
\end{equation}
which demonstrates that $\tilde{p}_1$ is a stream function for the leading-order horizontal velocity field. Based on \cref{eq_geostrophy}, one may invoke the definition of the vertical vorticity $\zeta$ to write
\begin{equation}
    \tilde{\zeta}_0 \coloneqq \hat{\bm{Z}} \cdot \left(\nabla_{h}\times\tilde{\bm{u}}_{h\, 0}\right)=\nabla_h^2 \,\tilde{p}_1,
\end{equation}
which coincides with \cref{eq_p_originalscale}, and further simplify the horizontal advection operator as
\begin{equation}\label{eq_uhnablah}
    \tilde{\bm{u}}_{h\, 0}\cdot\nabla_{h}\left(\cdot\right)=J\left[\,\tilde{p}_1,\,\cdot\,\right].
\end{equation}
Using \cref{eq_uhnablah}, we rearrange the vertical component of \cref{eq_fluc_u_0}, as well as \cref{eq_fluc_Theta_1} and \cref{eq_mean_Theta_2} with \cref{eq_theta0} substituted, to obtain eqs. \eqref{eq_w_originalscale}, \eqref{eq_theta_originalscale}, and \eqref{eq_Theta_originalscale}, respectively.

Furthermore, with \cref{eq_uhnablah} in place, taking the vertical component of the curl of \cref{eq_fluc_u_0} yields the following evolution equation for $\tilde{\zeta}_0$,
\begin{equation}\label{eq_derive_zeta}
    \partial_t \tilde{\zeta}_0+J\left[\tilde{p}_1,\,\tilde{\zeta}_0\right]+\nabla_h \cdot\tilde{\bm{u}}_{h\, 1}=\nabla_{h}^2 \,\tilde{\zeta}_0.
\end{equation}
On the other hand, the expansion of fluctuating continuity equation \eqref{eq_NS_fluc_w} gives
\begin{subequations}
    \begin{gather}
        \nabla_h\cdot\tilde{\bm{u}}_{h\, 0}=0,\label{eq_fluc_w_0}\\
       \nabla_h \cdot\tilde{\bm{u}}_{h\, n}+\partial_Z \,\tilde{w}_{n-1}=0,\quad  n\geq1,\label{eq_fluc_w_n}
    \end{gather}
\end{subequations}
with \cref{eq_fluc_w_0} imposing a consistent restriction on $\tilde{\bm{u}}_{h\,0}$ to \cref{eq_geostrophy}. According to \cref{eq_fluc_w_n}, a non-zero vertical gradient of $\tilde{w}_0$ does not violate the incompressibility condition, as mass conservation is ensured through the advection of $\tilde{\bm u}_{h\,1}$. We then use \cref{eq_fluc_w_n} to substitute $\nabla_h \cdot\tilde{\bm{u}}_{h\, 1}$ with $-\partial_Z \,\tilde{w}_{0}$ in \cref{eq_derive_zeta} to obtain \cref{eq_zeta_originalscale}, and all sub-equations of the system \eqref{eq_QG_viscousscale} have now been derived. Although the closed system \eqref{eq_QG_viscousscale} solely describes dynamics of the leading-order fields, the asymptotic expansion of the evolution equations in \eqref{eq_NS_mean} and \eqref{eq_NS_fluc} contains contributions at all orders, from which higher-order fields can be recovered.

   

\subsection{Remarks}

    The vertical component of \cref{eq_mean_u_minus1} reveals the hydrostatic balance
    \begin{equation}
        \partial_Z \,\overline{p}_{0} =\frac{{\mathcal{R}}}{Pr} {\overline{T}_{0}}
    \end{equation}
    attained between the leading-order mean pressure and buoyancy fields. Meanwhile, \eqref{eq_geostrophy} implies the geostrophic balance between the leading-order Coriolis force and horizontal pressure gradients, capturing the dominant force balance characteristic of the QG regimes. Compared to the classical geostrophic theory \citep{taylor1923experiments}, in \eqref{eq_QG_viscousscale} higher-order force balances are adjusted by the vertical convection, which renders the vertical pressure gradient \emph{non-hydrostatic} \citep{julien1998new} and the horizontal gradient \emph{quasi-geostrophic}. The balance \eqref{eq_geostrophy} does not account for boundary-layer effects, hence stress-free boundary conditions must be applied to ensure the physical consistency. For no-slip boundaries, Ekman boundary layers form and influence the bulk flow through Ekman pumping. Such effects can be parametrized as a modified mechanical boundary condition, following the approach of \citet{julien2016nonlinear}.
    
    Although the non-dimensionalisation itself involves no assumptions, the condition \eqref{eq_asymptoticexpansion} imposes that the chosen velocity, temperature, and pressure scales represent the largest amplitudes in powers of $\epsilon$, so that after non-dimensionalization, as $\epsilon\to 0$, no dynamics of order $O(\epsilon^{n})$ with $n<0$ arises. The chosen velocity scale filters out fast inertial waves. For the temperature scale, since rapid rotation tends to suppress convection, the conductive temperature scale is a natural choice. Finally, the geostrophic and hydrostatic balances satisfied by the leading-order pressure fields provide justification for the adopted pressure unit.

    The other assumption $\epsilon=E^{\frac{1}{3}}$, originally introduced by \citet{julien1998new}, is motivated by the $E^{\frac{1}{3}}$-scaling of the characteristic horizontal length scale in rotating Rayleigh–Bénard convection near the onset of convection \citep{chandrasekhar2013hydrodynamic}. The same scaling has been proven to hold asymptotically for rotating internally heated convection \citep{roberts1965thermal,jones2000onset,arslan2024internally}. Whether this scaling persists in the geostrophic turbulence regime remains an open question. Nevertheless, simulations of the rescaled incompressible Navier–Stokes equations \citep{julien2025rescaled,van2025bridging} on rotating Rayleigh–Bénard convection, which employ the same anisotropic non-dimensionalisation as the NH-QG equations but do not assume this scaling to be characteristic, demonstrate convergence toward NH-QG solutions as $E$ decreases.

\section{Boundedness of the auxiliary functional $\mathcal{V}$ along trajectories}
\label{Appendix_V}

\begin{proposition}[boundedness of the auxiliary functional along trajectories] Let 
\begin{equation}\label{eq_asV_auxiliaryfunction}
    \mathcal{V}\left\{\Theta,\,\theta\right\}= \frac{\beta}{2}\left\langle\theta^2+E^{-\frac{2}{3}}\Theta^2\right\rangle_{V}-E^{-\frac{2}{3}}\left\langle\phi (Z)\,\Theta\right\rangle_{V},
\end{equation}
in which $\beta$ is a positive parameter, and $\phi:[0,\,1]\to\mathbb{R}$ is square-integrable. Then, for any fixed choice of $\phi$ and $\beta$, and for any fixed value of $E$, $\mathcal{V}$ remains uniformly bounded in time when evaluated along trajectories of \cref{eq_QG}, subject to the stress-free, isothermal boundary conditions \eqref{eq_BC}, and regular initial data with finite $L^2$ norms over the domain.
\end{proposition}

\begin{proof}
   The argument begins with the evolution equation \eqref{eq_asV_energyevo} for the total temperature energy,
   \begin{equation}
    \frac{E^{-\frac{2}{3}}}{2} \frac{{\rm d}}{{\rm d} t}\langle \Theta^2+E^{\frac{2}{3}}\theta^2\rangle_{V}=-\langle \left|\nabla_h \theta\right|^2+(\partial_Z \Theta)^2\rangle_{V}+\langle\Theta\rangle_{V},
\end{equation}
According to the decomposition \eqref{eq_totaltemperature}, as $\overline{\theta}=0$, we have
\begin{equation}\label{eq_asV_energy}
    \langle \Theta^2+E^{\frac{2}{3}}\theta^2\rangle_{V}=\langle T^2\rangle_{V}.
\end{equation}

In \cref{eq_asV_energyevo}, as a consequence of horizontal periodicity, the horizontal dissipation satisfies
 \begin{equation}\label{eq_asV_dissitheta}
     \langle \left|\nabla_h \theta\right|^2\rangle_{V}\geq 4 \pi^2 E^{\frac{2}{3}} \min\left\{\frac{1}{L_x^2},\,\frac{1}{L_y^2}\right\}\,\langle\theta^2\rangle_{V}.
 \end{equation}
On the other hand, the isothermal boundary conditions together with the Poincaré inequality imply
 \begin{equation}\label{eq_asV_dissiTheta}
     \langle (\partial_Z \Theta)^2\rangle_{V}\geq \pi^2 \langle\Theta^2\rangle_{V}.
 \end{equation}
Combining \eqref{eq_asV_dissitheta} and \eqref{eq_asV_dissiTheta}, one may bound the volumetric thermal diffusion rate by
 \begin{equation}\label{eq_asV_dissi}
     \langle \left|\nabla_h \theta\right|^2+(\partial_Z \Theta)^2\rangle_{V} \geq \mu^2 \langle T^2\rangle_{V}, 
 \end{equation}
 in which 
 \begin{equation}
      \mu\coloneqq \pi \min\left\{\frac{2E^{1/3}}{L_x},\,\frac{2E^{1/3}}{L_y},\,1\right\}.
 \end{equation}
Additionally, by the Cauchy–Schwarz inequality, the source term in \cref{eq_asV_energyevo} satisfies
 \begin{equation}\label{eq_asV_source}
     \left|\langle \Theta\rangle_{V}\right|\leq \sqrt{\langle T^2 \rangle_{V}}.
 \end{equation}

Substituting \eqref{eq_asV_energy}, \eqref{eq_asV_dissi}, and \eqref{eq_asV_source} into \cref{eq_asV_energyevo} and dividing by $\sqrt{\langle T^2 \rangle_{V}}$ leads to
 \begin{equation}
    E^{-\frac{2}{3}}\frac{{\rm d}}{{\rm d} t}\sqrt{\langle T^2 \rangle_{V}}\leq -\mu^2 \sqrt{\langle T^2 \rangle_{V}}+1.
 \end{equation}
 Then, by Grönwall’s inequality,
 \begin{equation}\label{eq_asV_limsup}
     \sqrt{\langle T^2 \rangle_{V}}(t)\leq \sqrt{\langle T^2 \rangle_{V}}(0)\mathrm{e}^{-\mu^2 E^{\frac{2}{3}}t}+\frac{1}{\mu^2}\left(1-\mathrm{e}^{-\mu^2 E^{\frac{2}{3}}t}\right)< \sqrt{\langle T^2 \rangle_{V}}(0)+\frac{1}{\mu^2},\quad \forall t>0.
 \end{equation}

Moreover, the Cauchy–Schwarz inequality also gives
 \begin{equation}\label{eq_asV_phiTheta}
     \left|\langle \phi (Z)\Theta\rangle_{V}\right|\leq \|\phi(Z)\|_2 \sqrt{\langle T^2 \rangle_{V}}.
 \end{equation}
 By substituting \eqref{eq_asV_energy} and \eqref{eq_asV_phiTheta} into the auxiliary functional \eqref{eq_asV_auxiliaryfunction}, and combining with \cref{eq_asV_limsup}, we obtain the upper bound,
\begin{equation}
    \left| {\mathcal{V}}\right|< E^{-\frac{2}{3}}\left[\frac{ \beta}{2}\left(\sqrt{\langle T^2 \rangle_{V}}(0)+\frac{1}{\mu^2}\right)^2+\|\phi(Z)\|_2 \,\left(\sqrt{\langle T^2 \rangle_{V}}(0)+\frac{1}{\mu^2}\right)\right], \quad \forall t>0,
\end{equation}
which concludes the proof under the regularity of initial conditions.

\end{proof}

\section{Linear stability analysis}
\label{appendix_LST}

In this appendix, we perform a linear stability analysis  of the conduction state \eqref{eq_conduction} in the limit of infinite $Pr$. Specifically, the following result is demonstrated.
\begin{proposition}[Linear onset of convection, infinite $Pr$] Consider the convection system governed by \cref{eq_QG}, subject to the boundary conditions \eqref{eq_BC} and horizontal periodicity. Suppose that convective instability onsets in the form of stationary cellular convection, 
then the flow is linearly unstable for
\begin{equation}
    {\mathcal{R}}>{\mathcal{R}}_c=71.75.
\end{equation}
\end{proposition}

\begin{proof}
The governing equations are defined in \cref{eq_QG}. The steady conductive solution corresponds to $\Theta_c=\frac12 Z(1-Z)$  and $w=\psi=\theta=0$. To determine the criterion for linear instability, we find the smallest $\mathcal{R}$ for which linear disturbances grow. The linearised governing equations are given by
\begin{subequations}
\begin{gather}
    \partial_Z \psi = \mathcal{R}\theta + \nabla_h^2 w, \\
    -\partial_Z w = \nabla_h^2 \zeta, \label{eq_LST_zeta}\\
    \zeta = \nabla_h^2 \psi,\label{eq_LST_psi}
    \\
    \partial_t \theta + w  \,\Theta_c' = \nabla_h^2 \,\theta.\label{eq_theta_LST}
\end{gather}
\end{subequations}
Assuming the principle of exchange of stabilities, valid for large $Pr$, we consider steady onset modes.
Then, rearranging the equations in terms of $w$ gives
\begin{equation}
    -\partial_Z^2 w = \mathcal{R} \nabla_h^2(w \Theta'_c) + \nabla_h^6 w. \label{eq:comb_linear}
\end{equation}
Introducing in \cref{eq:comb_linear} the standard normal-mode ans\"{a}tz for the plane layer geometry of
\begin{equation}
    w = \sum_{\bm{k}} {w}_{\bm{k}}(Z) \, \mathrm{e}^{{\rm i}\left(k_x x+k_y y\right)},\quad  \theta = \sum_{\bm{k}} {\theta}_{\bm{k}}(Z) \, \mathrm{e}^{{\rm i}\left(k_x x+k_y y\right)},
\end{equation}
gives
\begin{equation}\label{eq_LST_w_infinitePr}
    w_{\bm{k}}''+\left[{\mathcal{R}}{k^2}\left(Z-\frac{1}{2}\right)-k^6\right]w_{\bm{k}}=0,
\end{equation}
in which $k=\sqrt{k_x^2+k_y^2}\in \mathbb{R}^+$. The solutions to \eqref{eq_LST_w_infinitePr} are given by Airy functions, since if we make the change of variables 
\begin{equation}
    \xi=-{\mathcal{R}}^{\frac{1}{3}}k^{\frac{2}{3}}\left(Z-\frac{1}{2}\right)+{\mathcal{R}}^{-\frac{2}{3}}k^{\frac{14}{3}},\quad W(\xi)= {w}_{\bm{k}}(Z),
\end{equation}
then \cref{eq_LST_w_infinitePr} reduces to the standard Airy equation,
\begin{equation}
    \frac{{\rm d}^2 W}{{\rm d}\xi ^2}-\xi\, W=0.
\end{equation}
The general solution of which is $ W(\xi)=C_1 {\rm Ai}(\xi)+C_2 {\rm Bi}(\xi)$, with $C_1$, $C_2$ constants, and $\rm Ai$, $\rm Bi$ denoting the Airy functions of the first and second kind, respectively. Imposing the boundary conditions ${w}_{\bm{k}}(0)={w}_{\bm{k}}(1)=0$ and requiring non-trivial solutions leads to the eigenvalue condition, 
\begin{equation}\label{eq_LST_infinitePr_eigen}
\begin{aligned}
    &{\rm Ai}\left({\mathcal{R}}^{-\frac{2}{3}}k^{\frac{14}{3}}+\frac{1}{2} {\mathcal{R}}^{\frac{1}{3}}k^{\frac{2}{3}}\right){\rm Bi}\left({\mathcal{R}}^{-\frac{2}{3}}k^{\frac{14}{3}}-\frac{1}{2} {\mathcal{R}}^{\frac{1}{3}}k^{\frac{2}{3}}\right)\\
    -&{\rm Ai}\left({\mathcal{R}}^{-\frac{2}{3}}k^{\frac{14}{3}}-\frac{1}{2} {\mathcal{R}}^{\frac{1}{3}}k^{\frac{2}{3}}\right){\rm Bi}\left({\mathcal{R}}^{-\frac{2}{3}}k^{\frac{14}{3}}+\frac{1}{2} {\mathcal{R}}^{\frac{1}{3}}k^{\frac{2}{3}}\right)=0.
\end{aligned}
\end{equation}
Based on \eqref{eq_LST_w_infinitePr}, $w''_{\bm{k}}=0$ wherever $w_{\bm{k}}=0$, and taking the $Z$-derivative of \cref{eq_LST_zeta,eq_LST_psi} further implies that $\partial_Z\zeta=\partial_Z\psi=0$. Therefore, the stress-free boundary condition is automatically satisfied. The linear critical reduced Rayleigh number is then given by
\begin{equation}
    {\mathcal{R}}_{c}\coloneqq \min\left\{{\mathcal{R}}\,\left|\,\text{eq.}\,\eqref{eq_LST_infinitePr_eigen}\, \&\, k>0\right.\right\}.
\end{equation}

For a given positive $k$, we numerically solve \cref{eq_LST_infinitePr_eigen} for ${\mathcal{R}}$ using the MATLAB \textit{vpasolve} function; and whenever solutions exist, only the smallest root is retained. By varying $k$, this procedure yields a spectrum ${\mathcal{R}}(k)$, as presented in figure \ref{fig_LST_infinit_sp}(a). The minimum of this curve identifies ${\mathcal{R}}_c = 71.75$, and the corresponding most unstable wavenumber $k_c = 1.505$. Figure \ref{fig_LST_infinit_sp}(b) displays the most unstable velocity and temperature modes. The constants $C_1$ and $C_2$ are fixed as $1.867$ and $5.783\times10^{-4}$, such that the boundary conditions are satisfied and the maximum bulk value of $w_{\bm{k}}$ is normalized to unity. The amplitudes of the corresponding $w$ and $\theta$ over one horizontal period is illustrated in figure \ref{fig_LST_infinit_sp}(c)\&(d), with $k_x=k_c$ and $k_y=0$. 

\begin{figure}
    \centering
    \includegraphics[width=0.7\linewidth]{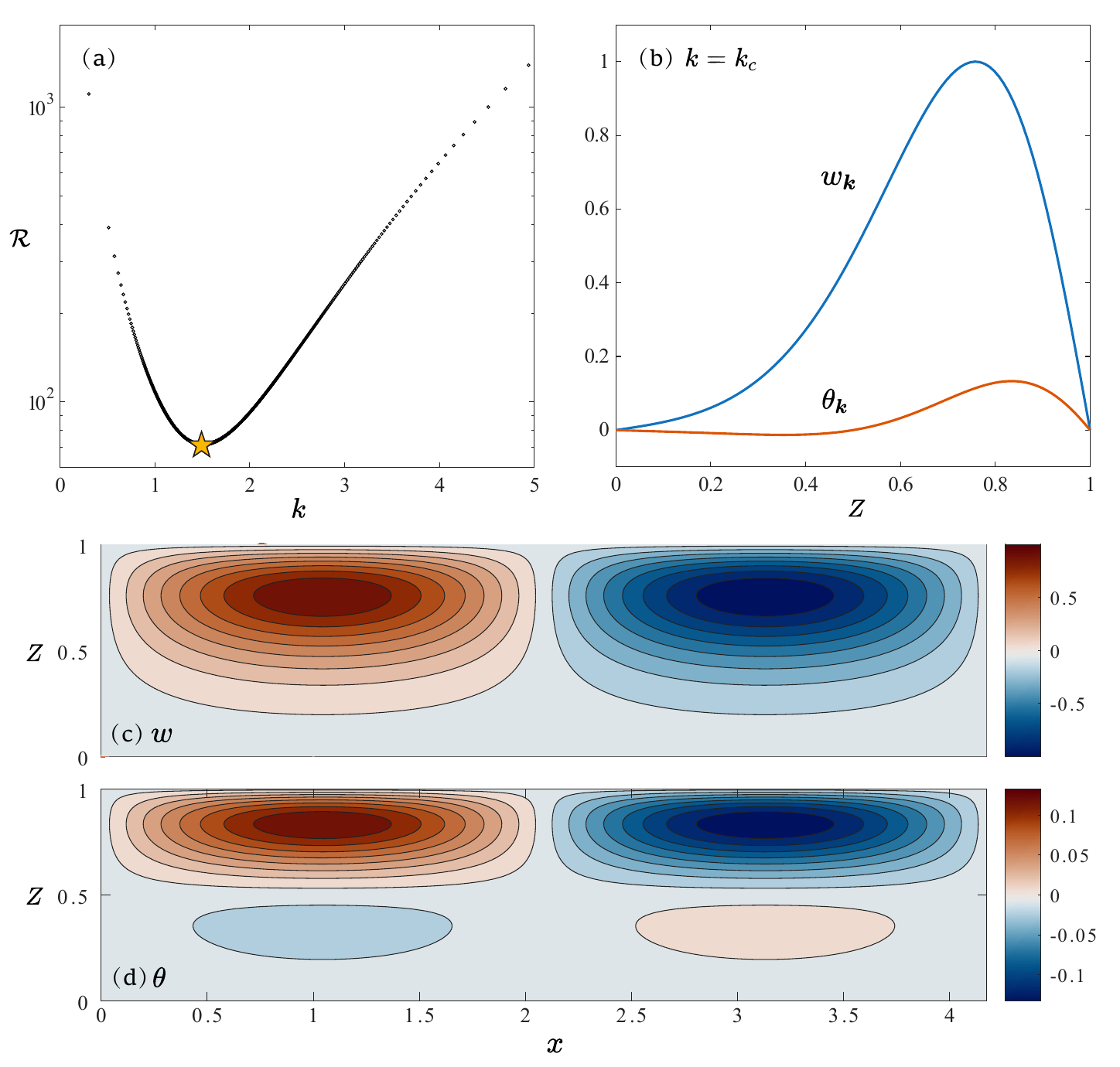}
    \caption{(a) Smallest root of \cref{eq_LST_infinitePr_eigen} for each given $k$. The star marks the $(k,{\mathcal{R}})$ pair corresponding to marginal linear stability. (b) Velocity and temperature modes at marginal linear stability, with $C_1$ and $C_2$ set to $1.867$ and $5.783\times10^{-4}$ respectively. (c) and (d) Amplitudes of the velocity and temperature modes in one horizontal period.} 
    \label{fig_LST_infinit_sp}
\end{figure}

\end{proof}
\end{appen}\clearpage

\bibliographystyle{jfm}
\bibliography{0_ref}

\end{document}